\documentclass[a4paper]{article}

\usepackage[utf8]{inputenc}
\usepackage[T1]{fontenc}
\usepackage{mathtools}
\usepackage{enumerate}
\usepackage{thmtools} 
\usepackage[inline]{enumitem}
\usepackage{tikz}
\usepackage{pdflscape}
\usepackage{pgfplots}
\usepackage{todonotes}
\usepackage{amssymb}
\usepackage{amsmath} 
\usepackage{amsthm} 
\usepackage[left=30mm,right=30mm,top=30mm,bottom=30mm]{geometry} 

\usepackage{etoolbox}
\usepackage{booktabs}
\usepackage{url}
\usepackage{xparse}
\usepackage{xspace}
\usepackage{subcaption}
\usepackage{array}
\usepackage[binary-units]{siunitx}
\usepackage{xparse}

\usepackage{etoolbox}

\DeclareDocumentCommand\R{}{\mathbb{R}}
\DeclareDocumentCommand\Z{}{\mathbb{Z}}
\DeclareDocumentCommand\conv{o}{\operatorname{conv}\IfValueTF{#1}{\left(#1\right)}{}}
\DeclareDocumentCommand\orderO{o}{\ensuremath{\mathcal{O}\IfValueTF{#1}{\left(#1\right)}{}}}
\DeclareDocumentCommand\orderTheta{o}{\ensuremath{\Theta\IfValueTF{#1}{\left(#1\right)}}}
\DeclareDocumentCommand\cplxNP{}{\mathsf{NP}}

\usepackage[colorlinks,citecolor=blue,linkcolor=blue,urlcolor=blue,bookmarks=false]{hyperref}
\usepackage[capitalize,noabbrev]{cleveref}

\newtheorem{theorem}{Theorem} 
\newtheorem{mylemma}[theorem]{Lemma}
\newtheorem{myconjecture}[theorem]{Conjecture}
\newtheorem{mycorollary}[theorem]{Corollary}
\newtheorem{myproposition}[theorem]{Proposition}
\newtheorem{mydefinition}[theorem]{Definition}
\newtheorem{myexample}[theorem]{Example}

\crefname{mylemma}{Lemma}{Lemmas}

\title{Simple odd $\beta$-cycle inequalities for binary polynomial optimization}
\author{Alberto Del Pia\footnote{Department of Industrial and Systems Engineering \& Wisconsin Institute for Discovery, University of Wisconsin-Madison, Madison, USA, \texttt{delpia@wisc.edu}} \and Matthias Walter\footnote{Department of Applied Mathematics, University of Twente, The Netherlands, \texttt{m.walter@utwente.nl}}}

\date{\small\today}

\DeclareDocumentCommand\Pml{}{\mathrm{ML}}
\DeclareDocumentCommand\Psr{}{\mathrm{SR}}
\DeclareDocumentCommand\Pfr{}{\mathrm{FR}}
\DeclareDocumentCommand\PfrOne{}{\mathrm{FR}_1}
\DeclareDocumentCommand\Pcr{}{\mathrm{CR}}
\DeclareDocumentCommand\boundSCIP{}{\mathrm{SCIP}}

\begin{document}

\maketitle

\begin{abstract}
  We consider the multilinear polytope which arises naturally in binary polynomial optimization.
  Del Pia and Di Gregorio introduced the class of odd $\beta$-cycle inequalities valid for this polytope, showed that these generally have Chv{\'a}tal rank~2 with respect to the standard relaxation and that, together with flower inequalities, they yield a perfect formulation for cycle hypergraph instances.
  Moreover, they describe a separation algorithm in case the instance is a cycle hypergraph.
  We introduce a weaker version, called simple odd $\beta$-cycle inequalities, for which we establish a strongly polynomial-time separation algorithm for arbitrary instances.
  These inequalities still have Chv{\'a}tal rank~2 in general and still suffice to describe the multilinear polytope for cycle hypergraphs.
  Finally, we report about computational results of our prototype implementation.
  The simple odd $\beta$-cycle inequalities sometimes help to close more of the integrality gap in the experiments; however, the preliminary implementation has substantial computational cost, suggesting room for improvement in the separation algorithm.
\end{abstract}

\DeclareDocumentCommand\U{}{\mathcal{U}}

\section{Introduction}
\label{sec intro}


In \emph{binary polynomial optimization} our task is to find a binary vector that maximizes a given multivariate polynomial function.
In order to give a mathematical formulation, it is useful to use a hypergraph $G=(V,E)$, where the node set $V$ represents the variables in the polynomial function, and the edge set $E$ represents the monomials with nonzero coefficients.
In a binary polynomial optimization problem, we are then given a hypergraph $G=(V,E)$, a profit vector $p \in \R^{V \cup E}$, and our goal is to solve the optimization problem
\begin{equation}
  \label{prob_bpo}
  \begin{aligned}
    \text{max } \left\{ \sum_{v \in V} p_v z_v + \sum_{e \in E} p_e \prod_{v \in e} z_v  :  
 z \in \{0,1\}^V \right\}.
  \end{aligned}
\end{equation}
Using Fortet's linearization~\cite{Fortet60,GloverW74}, we introduce binary auxiliary variables $z_e$, for $e \in E$, which are linked to the variables $z_v$, for $v \in V$, via the linear inequalities 
\begin{subequations}
  \label{eq_lin}
  \begin{align}
     z_v - z_e &\geq 0 &\quad& \forall e \in E, \ \forall v \in e \label{eq_lin_1} \\
     (z_e - 1) + \sum_{v \in e} (1-z_v) &\geq 0 &\quad& \forall e \in E \label{eq_lin_2}.
\end{align}
\end{subequations}
It is simple to see that 
\begin{align*}
\left\{ z \in \{0,1\}^{V \cup E} : z_e = \prod\limits_{v \in e} z_v ~\forall e \in E \right\}
=
\left\{ z \in \{0,1\}^{V \cup E}   :   
  \eqref{eq_lin} \right\}.
  \end{align*}
Hence, we can reformulate~\eqref{prob_bpo} as the integer linear optimization problem
\begin{equation}
  \label{prob_bpo_ip}
  \begin{aligned}
    \text{max } \left\{ \sum_{v \in V} p_v z_v + \sum_{e \in E} p_e z_e 
     : 
      \eqref{eq_lin}, \ 
     z \in \{0,1\}^{V \cup E} \right\}.
  \end{aligned}
\end{equation}
We define the \emph{multilinear polytope} $\Pml(G)$ \cite{DelPiaK17}, which is the convex hull of the feasible points of \eqref{prob_bpo_ip}, and its \emph{standard relaxation} $\Psr(G)$:
  \begin{align*}
  \Pml(G) & := \conv \left\{ z \in \{0,1\}^{V \cup E}   :   
  \eqref{eq_lin}
   \right\}, \\
  \Psr(G) & := \left\{ z \in [0,1]^{V \cup E}   :   
  \eqref{eq_lin}
   \right\}.
  \end{align*}

Recently, several classes of inequalities valid for $\Pml(G)$ have been introduced, including $2$-link inequalities \cite{CramaR17}, flower inequalities \cite{DelPiaK18}, running intersection inequalities~\cite{DelPiaK21}, and odd $\beta$-cycle inequalities~\cite{DelPiaD21}.
On a theoretical level, these inequalities fully describe the multilinear polytope for several hypergraph instances: 
flower inequalities for $\gamma$-acyclic hypergraphs, running intersection inequalities for kite-free $\beta$-acyclic hypergraphs, and flower inequalities together with odd $\beta$-cycle inequalities for cycle hypergraphs.
Furthermore, these cutting planes greatly reduce the integrality gap of \eqref{prob_bpo_ip}~\cite{DelPiaK21,DelPiaD21} and their addition leads to a significant reduction of the runtime of the state-of-the-art solver BARON \cite{DelPiaKS20}.
Unfortunately, we are not able to separate efficiently over most of these inequalities. 
In fact, while the simplest $2$-link inequalities can be trivially separated in polynomial time, there is no known polynomial-time algorithm to separate the other classes of cutting planes, and it is known that separating flower inequalities is $\cplxNP$-hard~\cite{DelPiaKS20}.

\paragraph{Contribution.}
In this paper we introduce a novel class of cutting planes called \emph{simple odd $\beta$-cycle inequalities}. 
As the name suggests, these inequalities form a subclass of the odd $\beta$-cycle inequalities introduced in~\cite{DelPiaD21}.
The main result of this paper is that simple odd $\beta$-cycle inequalities can be separated in strongly polynomial time.
While our inequalities form a subclass of the inequalities introduced in~\cite{DelPiaD21}, they still inherit the two most interesting properties of the odd $\beta$-cycle inequalities.
First, simple odd $\beta$-cycle inequalities can have Chv\'atal rank~$2$.
To the best of our knowledge, our algorithm is the first known polynomial-time separation algorithm over an exponential class of inequalities with  Chv\'atal rank~$2$. 
Second, simple odd $\beta$-cycle inequalities, together with standard linearization inequalities and flower inequalities with at most two neighbors, provide a perfect formulation of the multilinear polytope for cycle hypergraphs.
While the above results were first presented in the preliminary IPCO version of this paper \cite{dPWal22IPCO}, in this paper we also discuss redundancy and computations.
In fact, we provide computational evidence showing that our separation algorithm can lead to significant speedups in solving several applications that can be formulated as~\eqref{prob_bpo} with a hypergraph that contains $\beta$-cycles.
In particular, we present experiments on the image restoration problem in computer vision \cite{CramaR17,DelPiaD21}, and the low auto-correlation binary sequence problem in theoretical physics~\cite{Bernasconi87,MertensB98,DelPiaD21,POLIP14,MINLPLib20}.

\paragraph{Outline.}
We first introduce certain simple inequalities in \cref{sec_building_block} that are then combined to form the simple odd $\beta$-cycle inequalities in \cref{sec_inequalities}.
\cref{sec_separation} is dedicated to the polynomial-time separation algorithm.
In \cref{sec_redundancy}, we briefly address the question of redundancy, since our inequalities are formally defined for a more general structure than a $\beta$-cycle.
In \cref{sec_nonsimple}, we relate the simple odd $\beta$-cycle inequalities to the general (non-simple) odd $\beta$-cycle inequalities in~\cite{DelPiaD21}.
In \cref{sec_computations}, we discuss our implementation of the separation algorithm and we present computational results.
Finally, in \cref{sec_conclusions}, we discuss the computational results, and draw our conclusions.

\section{Building block inequalities}
\label{sec_building_block}

We consider certain affine linear functions $s : \R^{V \cup E} \to \R$ defined as follows.
For each $e \in E$ and each $v \in e$ we define
\begin{equation}
  s^{\mathrm{inc}}_{e,v}(z) \coloneqq z_v - z_e \tag{\ensuremath{s^{\mathrm{inc}}_{e,v}}} . \label{eq_block_incident}
\end{equation}
For each $e \in E$ and all $U,W \subseteq e$ with $U,W \neq \varnothing$ and $U \cap W = \varnothing$ we define
\begin{equation}
  s^{\mathrm{odd}}_{e,U,W}(z) \coloneqq 2z_e - 1 + \sum_{u \in U} (1-z_u) + \sum_{w \in W} (1-z_w) + \hspace{-3.5ex} \sum_{v \in e \setminus (U \cup W)} \hspace{-3ex} (2-2z_v) . \tag{\ensuremath{s^{\mathrm{odd}}_{e,U,W}}} \label{eq_block_odd}
\end{equation}
For all $e,f \in E$ with $e \cap f \neq \varnothing$ and all $U \subseteq e$ with $U \neq \varnothing$ and $U \cap f = \varnothing$ we define
\begin{equation}
  s^{\mathrm{one}}_{e,U,f}(z) \coloneqq 2z_e -1 + \sum_{u \in U} (1 - z_u) +(1 - z_f) + \hspace{-1em} \sum_{v \in e \setminus (U \cup f)} \hspace{-1em}  (2-2z_v) . \tag{\ensuremath{s^{\mathrm{one}}_{e,U,f}}} \label{eq_block_one}
\end{equation}
For all $e,f,g \in E$ with $e \cap f \neq \varnothing$, $e \cap g \neq \varnothing$ and $e \cap f \cap g = \varnothing$ we define
\begin{equation}
  s^{\mathrm{two}}_{e,f,g}(z) \coloneqq 2z_e -1 +(1- z_f) +(1- z_g) + \hspace{-1em} \sum_{v \in e \setminus (f \cup g)} \hspace{-1em}  (2-2z_v)  . \tag{\ensuremath{s^{\mathrm{two}}_{e,f,g}}} \label{eq_block_two}
\end{equation}

\DeclareDocumentCommand\sinc{mm}{\ensuremath{\text{\hyperref[eq_block_incident]{$s^{\mathrm{inc}}_{#1,#2}$}}}}
\DeclareDocumentCommand\sodd{mmm}{\ensuremath{\text{\hyperref[eq_block_odd]{$s^{\mathrm{odd}}_{#1,#2,#3}$}}}}
\DeclareDocumentCommand\sone{mmm}{\ensuremath{\text{\hyperref[eq_block_one]{$s^{\mathrm{one}}_{#1,#2,#3}$}}}}
\DeclareDocumentCommand\stwo{mmm}{\ensuremath{\text{\hyperref[eq_block_two]{$s^{\mathrm{two}}_{#1,#2,#3}$}}}}
In this paper we often refer to \ref{eq_block_incident}, \ref{eq_block_odd}, \ref{eq_block_one}, \ref{eq_block_two} as \emph{building blocks}.
Although in these definitions $U$ and $W$ can be arbitrary subsets of an edge $e$, in the following $U$ and $W$ will always correspond to the intersection of $e$ with another edge.
In the next lemma we will show that all building blocks are nonnegative on a relaxation of $\Pml(G)$ obtained by adding some flower inequalities~\cite{DelPiaK18} to $\Psr(G)$, which we will define now.
For ease of notation, in this paper, we denote by $[m]$ the set $\{1,\dots,m\}$, for any nonnegative integer $m$.

Let $f \in E$ and let $e_i$, $i \in [m]$, be a collection of pairwise distinct edges in $E$, adjacent to $f$, such that $f \cap e_i \cap e_j = \emptyset$ for all $i,j \in [m]$ with $i \neq j$.
Then the \emph{flower inequality} \cite{DelPiaK18,DelPiaD21} centered at $f$ with neighbors $e_i$, $i \in [m]$, is defined by
\begin{equation}
  (z_f - 1) + \sum_{i \in [m]} (1-z_{e_i}) + \sum_{v \in f \setminus \cup_{i \in [m]} e_i} (1-z_v) \geq 0. \label{eq_flower}
\end{equation}
We denote by $\Pfr(G)$ the polytope obtained from $\Psr(G)$ by adding all flower inequalities with at most two neighbors. 
Clearly $\Pfr(G)$ is a relaxation of $\Pml(G)$.
Furthermore, $\Pfr(G)$ is defined by a number of inequalities that is bounded by a polynomial in $|V|$ and $|E|$.

\begin{mylemma}
  \label{thm_slack_functions}
  Let $G = (V,E)$ be a hypergraph and let $s$ be one of~\ref{eq_block_incident}, \ref{eq_block_odd}, \ref{eq_block_one}, \ref{eq_block_two}.
  Then $s(z) \geq 0$ is valid for $\Pfr(G)$.
  Furthermore, if $z \in \Pml(G) \cap \Z^{V \cup E}$ and $s(z) = 0$, then the corresponding implication below holds.
  \begin{enumerate}[label=(\roman*)]
  \item
    \label{enum_inc}
    If $\ref{eq_block_incident}(z) = 0$ then $z_v = z_e$.
  \item
    \label{enum_odd}
    If $\ref{eq_block_odd}(z) = 0$ then $\prod_{u \in U} z_u + \prod_{w \in W} z_w = 1$.
  \item
    \label{enum_one}
    If $\ref{eq_block_one}(z) = 0$ then $z_f + \prod_{u \in U} z_u = 1$.
  \item
    \label{enum_two}
    If $\ref{eq_block_two}(z) = 0$ then $z_f + z_g = 1$.
  \end{enumerate}
\end{mylemma}

\begin{proof}
  First, $\ref{eq_block_incident}(z) \geq 0$ is part of the standard relaxation and implication~\ref{enum_inc} is obvious.

  Second, $\sodd{e}{U}{W}(z) \geq 0$ is the sum of the following inequalities from the standard relaxation:
  $z_e \geq 0$, $1-z_v \geq 0$ for all $v \in e \setminus (U \cup W)$, and $(z_e - 1) + \sum_{v \in e} (1-z_v) \geq 0$.
  If $z \in \Pml(G) \cap \Z^{V \cup E}$ and $\sodd{e}{U}{W}(z) = 0$, then each of these inequalities must be tight, thus $z_e = 0$, $z_v = 1$ for each $v \in e \setminus (U \cup W)$.
  The last (tight) inequality yields $-1 + \sum_{v \in U \cup W} (1-z_v) = 0$, i.e., precisely one variable $z_v$, for $v \in U \cup W$, is 0, while all others are 1,
  which yields implication~\ref{enum_odd}.

  Third, $\sone{e}{U}{f}(z) \geq 0$ is the sum of the following inequalities:
  $z_e \geq 0$, $1-z_v \geq 0$ for all $v \in e \setminus (U \cup f)$ and $(z_e-1) + (1-z_f) + \sum_{v \in e \setminus f} (1-z_v) \geq 0$.
  The latter is the flower inequality centered at $e$ with neighbor $f$.
  If $z \in \Pml(G) \cap \Z^{V \cup E}$ and $\sone{e}{U}{f}(z) = 0$, then each of these inequalities must be tight, thus $z_e = 0$, $z_v = 1$ for each $v \in e \setminus (U \cup f)$.
  The last (tight) inequality yields $-1 + (1-z_f) + \sum_{u \in U} (1-z_u) = 0$, i.e., either $z_f = 1$ and $z_u = 0$ for exactly one $u \in U$, or $z_f = 0$ and $z_u = 1$ holds for all $u \in U$.
  Both cases yield implication~\ref{enum_one}.

  Fourth, we consider $\stwo{e}{f}{g}(z) \geq 0$.
  Note that due to $e \cap f \neq \varnothing$, $e \cap g \neq \varnothing$ and $e \cap f \cap g = \varnothing$, the three edges $e,f,g$ must all be different.
  Thus, $\stwo{e}{f}{g}(z) \geq 0$ is the sum of $z_e \geq 0$, $1-z_v \ge 0$ for all $v \in e \setminus (f \cup g)$ and of $(z_e - 1) + (1 - z_f) + (1-z_g) + \sum_{v \in e \setminus (f \cup g)} (1-z_v) \geq 0$.
  The latter is the flower inequality centered at $e$ with neighbors $f$ and $g$.
  If $z \in \Pml(G) \cap \Z^{V \cup E}$ and $\stwo{e}{f}{g}(z) = 0$ holds, then each of the involved inequalities must be tight, thus $z_e = 0$ and $z_v = 1$ for each $v \in e \setminus (f \cup g)$.
  The last (tight) inequality implies $-1 + (1-z_f) + (1-z_g) = 0$, i.e., $z_f + z_g = 1$.
  Hence, implication~\ref{enum_two} holds.
  \qed
\end{proof}

\section{Simple odd $\beta$-cycle inequalities}
\label{sec_inequalities}

We will consider signed edges by associating either a ``$+$'' or a ``$-$'' with each edge.
We denote by $\{\pm\}$ the set $\{+,-\}$ and 
by $-p$ a sign change for $p \in \{\pm\}$.
In order to introduce simple odd $\beta$-cycle inequalities, we first present some more definitions.

\begin{mydefinition}
\label{def inequalities}
  A \emph{closed walk} in $G$ of length $k \geq 3$ is a sequence $C = v_1$-$e_1$-$v_2$-$e_2$-$v_3$-$\dotsb$-$v_{k-1}$-$e_{k-1}$-$v_k$-$e_k$-$v_1$, where we have $e_i \in E$ as well as $v_i \in e_{i-1} \cap e_i$ and $e_{i-1} \cap e_i \cap e_{i+1} = \varnothing$ for each $i \in [k]$, where we denote $e_0 \coloneqq e_k$ and $e_{k+1} \coloneqq e_1$ for convenience.
  A \emph{signature} of $C$ is a map $\sigma : [k] \to \{\pm\}$.
  A \emph{signed closed walk in $G$} is a pair $(C,\sigma)$ for a closed walk $C$ and a signature $\sigma$ of $C$.
  Similarly, we denote $v_0 \coloneqq v_k$, $v_{k+1} \coloneqq v_1$, $\sigma(0) \coloneqq \sigma(k)$ and $\sigma(k+1) \coloneqq \sigma(1)$.
  We say that $(C,\sigma)$ is \emph{odd} if there is an odd number of indices $i \in [k]$ with $\sigma(i) = -$; otherwise we say that $(C,\sigma)$ is \emph{even}.
  Finally, for any signed closed walk $(C,\sigma)$ in $G$, its \emph{length function} is the map $\ell_{(C,\sigma)} : \Pfr(G) \to \R$ defined by
  \begin{multline*}
    \ell_{(C,\sigma)}(z) \coloneqq
    \hspace{-2ex} \sum_{i \in I_{(+,+,+)}} \hspace{-2ex} \big(\sinc{e_i}{v_i}(z) + \sinc{e_i}{v_{i+1}}(z) \big)
    + \hspace{-2ex} \sum_{i \in I_{(-,-,-)}} \hspace{-2ex} \sodd{e_i}{e_i \cap e_{i-1}}{e_i \cap e_{i+1}}(z) \\
    + \hspace{-2ex} \sum_{i \in I_{(+,+,-)}} \hspace{-2ex} \sinc{e_i}{v_i}(z)
    + \hspace{-2ex} \sum_{i \in I_{(-,-,+)}} \hspace{-2ex}
    \sone{e_i}{e_i \cap e_{i-1}}{e_{i+1}}(z)
    + \hspace{-2ex} \sum_{i \in I_{(-,+,+)}} \hspace{-2ex} \sinc{e_i}{v_{i+1}}(z) \\
    + \hspace{-2ex} \sum_{i \in I_{(+,-,-)}} \hspace{-2ex} \sone{e_i}{e_i \cap e_{i+1}}{e_{i-1}}(z)
    + \hspace{-2ex} \sum_{i \in I_{(+,-,+)}} \hspace{-2ex} \stwo{e_i}{e_{i-1}}{e_{i+1}}(z),
  \end{multline*}
  where $I_{(a,b,c)}$ is the set of edge indices $i$ for which $e_{i-1}$, $e_i$ and $e_{i+1}$ have sign pattern $(a,b,c) \in \{\pm\}^3$, i.e., $I_{(a,b,c)} \coloneqq \{ i \in [k] : \sigma(i-1) = a,~ \sigma(i) = b,~ \sigma(i+1) = c \}$.
\end{mydefinition}

We remark that the definition of $\ell_{(C,\sigma)}(z)$ is independent of where the closed walk starts and ends.
Namely, if instead of $C$ we consider 
$C' = v_{i}$-$e_{i}$
-$\dotsb$-$v_{k}$-$e_{k}$-$v_{1}$-$e_{1}$-$\dotsb$-$v_{i-1}$-$e_{i-1}$-$v_i$, and we define $\sigma'$ accordingly, then we have $\ell_{(C,\sigma)}(z) = \ell_{(C',\sigma')}(z)$.
Moreover, if $\sigma(i-1) = -$ or $\sigma(i) = -$, then $\ell_{(C,\sigma)}(z)$ is independent of the choice of $v_i \in e_{i-1} \cap e_i$.

By \cref{thm_slack_functions}, the length function of a signed closed walk is nonnegative.
We will show that for odd signed closed walks, the length function evaluated in each integer solution is at least $1$.
Hence, we define the \emph{simple odd $\beta$-cycle inequality} corresponding to the odd signed closed walk $(C,\sigma)$ as
\begin{equation}
  \ell_{(C,\sigma)}(z) \geq 1. \label{eq_simple_odd_beta_cycle}
\end{equation}
We first establish that this inequality is indeed valid for $\Pml(G)$.
\begin{theorem}
  \label{thm_simple_odd_beta_valid}
  Simple odd $\beta$-cycle inequalities~\eqref{eq_simple_odd_beta_cycle} are valid for $\Pml(G)$.
\end{theorem}

\begin{proof}
  Let $z \in \Pml(G) \cap \{0,1\}^{V \cup E}$ and assume, for the sake of contradiction, that $z$ violates inequality~\eqref{eq_simple_odd_beta_cycle} for some odd signed closed walk $(C,\sigma)$.
  Since the coefficients of $\ell_{(C,\sigma)}$ are integer, we obtain $\ell_{(C,\sigma)} \leq 0$.
  From \cref{thm_slack_functions}, we have that $s(z) = 0$ holds for all involved functions $s(z)$.
  Moreover, edge variables $z_{e_i}$ for all edges $e_i$ with $\sigma(i) = +$, node variables $z_{v_i}$ for all nodes $v_i$ with $\sigma(i-1) = \sigma(i) = +$, and the expressions $\prod_{v \in e_{i-1} \cap e_i} z_v$ for all nodes $i$ with $\sigma(i-1) = \sigma(i) = -$ are either equal or complementary, where the latter happens if and only if the corresponding edge $e_i$ satisfies $\sigma(i) = -1$.
  Since the signed closed walk $C$ is odd, this yields a contradiction $z_e = 1-z_e$ for some edge $e$ of $C$ or $z_v = 1-z_v$ for some node $v$ of $C$ or $\prod_{v \in e \cap f} z_v = 1 - \prod_{v \in e \cap f} z_v$ for a pair $e,f$ of subsequent edges of $C$.
  \qed
\end{proof}

%

Next, we provide an example of a simple odd $\beta$-cycle inequality.

\begin{figure}[ht]
\centering
\includegraphics[width=.5\textwidth]{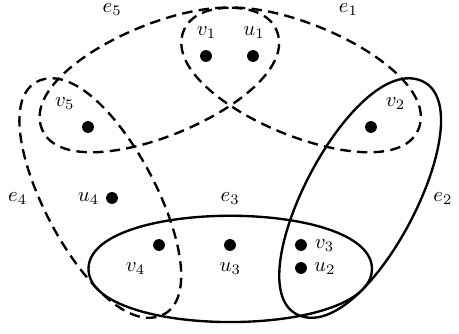}
\caption{Figure of the closed walk considered in Example~\ref{ex}. The solid edges have sign $+$ and the dashed edges have sign $-$.}
\label{fig_example}
\end{figure}

\begin{myexample}
\label{ex}
We consider the closed walk of length $5$ given by the sequence $C = v_1$-$e_1$-$v_2$-$e_2$-$v_3$-$\dotsb$-$v_{5}$-$e_{5}$-$v_1$ with signature 
$(\sigma(1), \sigma(2), \dotsc, \sigma(5)) = (-, +, +, -, -)$
depicted in \cref{fig_example}.
We have $1 \in I_{(-,-,+)}$, $2 \in I_{(-,+,+)}$, $3 \in I_{(+,+,-)}$, $4 \in I_{(+,-,-)}$, $5 \in I_{(-,-,-)}$.
The corresponding simple odd $\beta$-cycle inequality is $\ell_{(C,\sigma)}(z) \ge 1$.
Using Definition~\ref{def inequalities}, we write $\ell_{(C,\sigma)}(z)$ in terms of the building blocks as
  \begin{align*}
    \ell_{(C,\sigma)}(z) = \
    & \sone{e_1}{e_1 \cap e_{5}}{e_{2}}(z)
    + \sinc{e_2}{v_3}(z)
    + \sinc{e_3}{v_3}(z)
    + \sone{e_4}{e_4 \cap e_{5}}{e_{3}}(z)
    + \sodd{e_5}{e_5 \cap e_{4}}{e_5 \cap e_{1}}(z).
  \end{align*}
  Using the definition of the building blocks, we obtain 
    \begin{align*}
    \ell_{(C,\sigma)}(z) = \ 
    & + 2z_{e_1} -1 + \sum_{u \in {e_1 \cap e_{5}}} (1 - z_u) +(1 - z_{e_{2}}) + \hspace{-6mm}\sum_{v \in {e_1} \setminus ({e_1 \cap e_{5}} \cup {e_{2}})}\hspace{-6mm} (2-2z_v) \\
    & + (z_{v_{3}} - z_{e_{2}} ) + (z_{v_{{3}}} - z_{e_3}) \\
    & + 2z_{e_4} -1 + \sum_{u \in {e_4 \cap e_{5}}} (1 - z_u) +(1 - z_{e_{3}}) + \hspace{-6mm}\sum_{v \in {e_4} \setminus ({e_4 \cap e_{5}} \cup {e_{3}})}\hspace{-6mm} (2-2z_v) \\
    & + 2z_{e_5} - 1 + \sum_{u \in {e_5 \cap e_{4}}} (1-z_u) + \sum_{w \in {e_5 \cap e_{1}}} (1-z_w) + \hspace{-8mm}\sum_{v \in {e_5} \setminus ({e_5 \cap e_{4}} \cup ({e_5 \cap e_{1}}))}\hspace{-8mm} (2-2z_v).
  \end{align*}
We write the sums explicitly and obtain
  \begin{align*}
    \ell_{(C,\sigma)}(z) & = 
    + 2z_{e_1} -1 + (1 - z_{v_1}) + (1 - z_{u_1}) +(1 - z_{e_{2}}) \\
    & \hspace{.4cm} + (z_{v_{3}} - z_{e_{2}} ) + (z_{v_{{3}}} - z_{e_3}) \\
    & \hspace{.4cm} + 2z_{e_4} -1 + (1 - z_{v_5}) +(1 - z_{e_{3}}) +  (2-2z_{u_4}) \\
    & \hspace{.4cm} + 2z_{e_5} - 1 + (1-z_{v_5}) + (1-z_{v_1}) + (1-z_{u_1}) \\
    & = 2(z_{e_1}  - z_{e_{2}} - z_{e_3} + z_{e_4} + z_{e_5} 
    - z_{v_1} - z_{u_1} + z_{v_{{3}}} - z_{u_4} - z_{v_5}) + 7. \hspace{5mm} \diamond
  \end{align*}

\end{myexample}

Example~\ref{ex} suggests that, when the function is written explicitly, the coefficients in the function $\ell_{(C,\sigma)}(z)$ exhibit a certain pattern.
This different expression of $\ell_{(C,\sigma)}(z)$ is formalized in the next lemma.

\begin{mylemma}
  \label{thm_length_combined}
  Given a signed closed walk $(C,\sigma)$ in $G$ with $k \geq 3$, we have
  \begin{multline}
    \ell_{(C,\sigma)}(z)
    = \sum_{\substack{i \in [k] \\ \sigma(i) = -}} \hspace{-1ex} (2{z}_{e_i} + 1)
    - \hspace{-1ex} \sum_{\substack{i \in [k] \\ \sigma(i) = +}} \hspace{-1ex} 2{z}_{e_i}
    + \hspace{-3ex} \sum_{\substack{i \in [k] \\ \sigma(i-1) = \sigma(i) = +}} \hspace{-3ex} 2{z}_{v_i}
    + \hspace{-3ex} \sum_{\substack{i \in [k] \\ \sigma(i-1) = \sigma(i) = - \\ v \in e_{i-1} \cap e_i}} \hspace{-3ex} 2 (1 - {z}_v) \\
    + \hspace{-2ex} \sum_{\substack{i \in [k] : \sigma(i) = - \\ v \in e_i \setminus  (e_{i-1} \cup e_{i+1}) }} \hspace{-4ex} 2 (1 - {z}_v)
    -2 |\{ i \in [k] : \sigma(i-1) = \sigma(i) = - \}|.
    \label{eq_length_combined}
  \end{multline}
\end{mylemma}

Using \cref{thm_length_combined}, we obtain the following result.

\begin{myproposition}
  \label{thm_simple_odd_beta_cg}
  Simple odd $\beta$-cycle inequalities are Chv{\'a}tal-Gomory inequalities for $\Pfr(G)$ and can be written in the form
  \begin{align}
  \label{eq_simple_odd_beta_cg}
  \begin{split}
    \sum_{\substack{i \in [k] \\ \sigma(i) = -}} \hspace{-1ex} z_{e_i}
    - \hspace{-1ex} \sum_{\substack{i \in [k] \\ \sigma(i) = +}} \hspace{-1ex} z_{e_i}
    + \hspace{-3ex} \sum_{\substack{i \in [k] \\ \sigma(i-1) = \sigma(i) = +}} \hspace{-3ex} z_{v_i}
    - \hspace{-3ex} \sum_{\substack{i \in [k] \\ \sigma(i-1) = \sigma(i) = - \\ v \in e_{i-1} \cap e_i}} \hspace{-3ex} (z_v - 1)
    - \hspace{-2ex} \sum_{\substack{i \in [k] : \sigma(i) = - \\ v \in e_i \setminus  (e_{i-1} \cup e_{i+1}) }} \hspace{-4ex} (z_v - 1) \\
    \geq  \frac{1 - |\{ i \in [k] : \sigma(i) = - \}|}{2}  - |\{ i \in [k] : \sigma(i-1) = \sigma(i) = - \}|.
    \end{split}
  \end{align}
\end{myproposition}

\begin{proof}
  Let $(C,\sigma)$ be an odd signed closed walk in a hypergraph $G$.
  From \cref{thm_slack_functions} we obtain that $\ell_{(C,\sigma)}(z) \geq 0$ holds for each $z \in \Pfr(G)$.
  \cref{thm_length_combined} reveals that in the inequality $\ell_{(C,\sigma)}(z) \geq 0$, all variables' coefficients are even integers, while the constant term is an odd integer.
  Hence, the inequality divided by $2$ has integral variable coefficients, and we can obtain the corresponding Chv{\'a}tal-Gomory inequality by rounding the constant term up.
  The resulting inequality is the simple odd $\beta$-cycle inequality~\eqref{eq_simple_odd_beta_cycle} scaled by $1/2$ and has the form~\eqref{eq_simple_odd_beta_cg}.
  This shows that simple odd $\beta$-cycle inequalities are Chv{\'a}tal-Gomory inequalities for $\Pfr(G)$.
  \qed
\end{proof}

It follows from \cref{thm_simple_odd_beta_cg} that, under some conditions on $(C,\sigma)$, simple odd $\beta$-cycle inequalities are in fact $\{0,1/2\}$-cuts (see~\cite{CapraraF96}) with respect to $\Pfr(G)$.
Some classes of such cutting planes can be separated in polynomial time, in particular if the involved inequalities only have two odd coefficients.
In such a case, these inequalities are patched together such that odd coefficients cancel out and eventually all coefficients are even.
We want to emphasize that this generic separation approach does not work in our case since our building block inequalities may have more than 2 odd-degree coefficients.
Nevertheless, the separation algorithm presented in the next section is closely related to the idea of cancellation of odd-degree coefficients.

\section{Separation algorithm}
\label{sec_separation}

The main goal of this section is to show that the separation problem over simple odd $\beta$-cycle inequalities can be solved in strongly polynomial time (\cref{thm_simple_odd_beta_separation}).
This will be achieved by means of an auxiliary undirected graph in which several shortest-path computations must be carried out.
The auxiliary graph is inspired by the one for the separation problem of odd-cycle inequalities for the maximum cut problem~\cite{BarahonaM86}.
However, to deal with our different problem and the more general hypergraphs we will extend it significantly.

\DeclareDocumentCommand\triples{}{\ensuremath{\mathcal{T}}}

Let $G = (V,E)$ be a hypergraph and let $\hat{z} \in \Pfr(G)$.
Define $\triples{} \coloneqq \{ (e,f,g) \in E : e \cap f \neq \varnothing,~ f \cap g \neq \varnothing,~ e \cap f \cap g = \varnothing \}$ to be the set of potential subsequent edge triples.
We define the auxiliary graph
\[
  \bar{G} = (\bar{V}, \bar{E}) = (\bar{V}_+ \cup \bar{V}_- \cup \bar{V}_{\mathrm{E}}, \bar{E}^{-,-,-} \cup \bar{E}^{+,-,+} \cup \bar{E}^{+,-,-}\cup \bar{E}^{+,+,\pm})
\]
and length function $\bar{\ell} : \bar{E} \to \R$ as follows.
\begin{align*}
  \bar{V}_{+} \coloneqq
    &\; V \times \{ \pm \} \\
  \bar{V}_{-} \coloneqq
    &\; \{ e \cap f  : e,f \in E,~ e \neq f,~ e \cap f \neq \varnothing \} \times \{\pm\} \\
  \bar{V}_{\mathrm{E}} \coloneqq
    &\; E \times \{\pm\} \\
   \bar{E}^{-,-,-} \coloneqq
     &\;\{ \{(e \cap f,p),(f \cap g,-p)\} : (e,f,g) \in \triples{},~ p \in \{\pm\} \} \\
  \bar{\ell}_{\{(U,p),(W,-p)\}}  \coloneqq
    &\;\min\limits_{e,f,g} \{ \sodd{f}{U}{W}(\hat{z}) : U = e \cap f,~ W = f \cap g \text{ for some } (e,f,g) \in \triples{} \} \\
  \bar{E}^{+,-,+} \coloneqq
    &\;\{ \{(e,p),(g,-p)\} : e,g \in E, e \cap f \neq \varnothing \text{ and } f \cap g \neq \varnothing  \\
    &\; \text{ for some } f \in E \text{ with $e \cap f \cap g = \varnothing$}, \ p \in \{\pm\} \} \\
  \bar{\ell}_{\{(e,p),(g,-p)\}} \coloneqq
    &\;\min\limits_f \{ \stwo{e}{f}{g}(\hat{z}) : f \in E,~ e \cap f \neq \varnothing, ~ f \cap g \neq \varnothing,~ e \cap f \cap g = \varnothing \} \\
  \bar{E}^{+,-,-} \coloneqq
    &\;\{ \{(e,p),(f \cap g,-p)\} : (e,f,g) \in \triples{},~ p \in \{\pm\} \} \\
  \bar{\ell}_{\{(e,p),(U,-p)\}} \coloneqq
    &\;\min\limits_{f,g} \{ \sone{f}{U}{e}(\hat{z}) : (e,f,g) \in \triples{},~ U = f \cap g \} \\
  \bar{E}^{+,+,\pm} \coloneqq
    &\;\{ \{(v,p),(e,p)\} : v \in e \in E, \ p \in \{\pm\} \} \\
  \bar{\ell}_{\{(v,p),(e,p)\}} \coloneqq
    &\;\ref{eq_block_incident}(\hat{z})
\end{align*}

We point out that the graph $\bar G$ can have parallel edges, possibly with different lengths.
We immediately obtain the following corollary from \cref{thm_slack_functions}.

\begin{mycorollary}
The edge lengths $\bar{\ell} : \bar{E} \to \R$ are nonnegative.
\end{mycorollary}

We say that two nodes $\bar{u}, \bar{v} \in \bar{V}$ are \emph{twins} if they only differ in the second component, i.e., the sign.
We call a walk $\bar{W}$ in the graph $\bar{G}$ a \emph{twin walk} if its end nodes are twin nodes.
For a walk $\bar{W}$ in $\bar G$, we denote by $\bar{\ell}(\bar{W})$ the total \emph{length}, i.e., the sum of the edge lengths $\bar{\ell}_e$ along the edges $e$ in $\bar{W}$.
In the next two lemmas we study the relationship between odd signed closed walks in $G$ and twin walks in $\bar{G}$.

\begin{mylemma}
  \label{thm_hyperwalk_to_walk}
  For each odd signed closed walk $(C,\sigma)$ in $G$ there exists a twin walk $\bar{W}$ in $\bar{G}$ of length $\bar{\ell}(\bar{W}) \leq 1 + s$, where $s$ is the slack of the simple odd $\beta$-cycle inequality~\eqref{eq_simple_odd_beta_cycle} induced by $(C,\sigma)$ with respect to $\hat{z}$.
  In particular, if the inequality is violated by $\hat{z}$, then we have $\bar{\ell}(\bar{W}) < 1$.
\end{mylemma}

\begin{proof}
  Let $(C,\sigma)$ be an odd signed closed walk with $C = v_1$-$e_1$-$v_2$-$e_2$-$v_3$-$\dotsb$-$v_{k-1}$-$v_{k-1}$-$v_k$-$e_k$-$v_1$.
  For $i \in [k]$, let $p_i \coloneqq \prod_{j=1}^i \sigma(j)$ be the product of signs of all edges up to $e_i$.
  Moreover, define $p_0 \coloneqq \sigma(0) = \sigma(k)$.
  For each $i \in [k]$, we determine a walk $\bar{W}_i$ in $\bar G$ of length at most $2$, and construct $\bar{W}$ by going along all these walks in their respective order.
  The walk
  $\bar{W}_i$ depends on $\sigma(i-1)$, $\sigma(i)$ and $\sigma(i+1)$:
  \[
    \bar{W}_i \coloneqq \begin{cases}
      (v_i,p_{i-1}) \rightarrow (e_i,p_i) \rightarrow (v_{i+1},p_i) & \text{if } i \in I_{+,+,+} \\
      (v_i,p_{i-1}) \rightarrow (e_i,p_i)  & \text{if } i \in I_{+,+,-} \\
      (e_i,p_i) \rightarrow (v_{i+1},p_i)  & \text{if } i \in I_{-,+,+} \\
      (e_i,p_i) \hfill \text{(length $0$)} & \text{if } i \in I_{-,+,-} \\
      (e_{i-1},p_{i-1}) \rightarrow (e_i \cap e_{i+1},p_i) & \text{if } i \in I_{+,-,-} \\
      (e_{i-1} \cap e_i,p_{i-1}) \rightarrow (e_{i+1},p_i)  & \text{if } i \in I_{-,-,+} \\
      (e_{i-1} \cap e_i,p_{i-1}) \rightarrow (e_i \cap e_{i+1},p_i) & \text{if } i \in I_{-,-,-} \\
      (e_{i-1},p_{i-1}) \rightarrow (e_{i+1},p_i) & \text{if } i \in I_{+,-,+}.
                        \end{cases}
  \]

  The walks $\bar{W}_i$ help to understand the meaning of the different node types: the walk $\bar{W}_i$ starts at a node from $\bar{V}_+$ if $\sigma(i-1) = \sigma(i) = +$, it starts at a node from $\bar{V}_-$ if $\sigma(i-1) = \sigma(i) = -$, and it starts at a node from $\bar{V}_{\mathrm{E}}$ if $\sigma(i-1) \neq \sigma(i)$ holds.
  Similarly, the walk $\bar{W}_i$ ends at a node from $\bar{V}_+$ if $\sigma(i) = \sigma(i+1) = +$, it ends at a node from $\bar{V}_-$ if $\sigma(i) = \sigma(i+1) = -$, and it ends at a node from $\bar{V}_{\mathrm{E}}$ if $\sigma(i) \neq \sigma(i+1)$ holds. 

  Note that all edges traversed by each $\bar W_i$ are indeed in $\bar E$.
  It is easily verified that, for each $i \in [k-1]$, the walk $\bar{W}_i$ ends at the same node at which the walk $\bar{W}_{i+1}$ starts.
  Hence $\bar W$ is indeed a walk in $\bar G$.
  Since $v_{k+1} = v_1$ holds, $C$ is closed
  and $(C,\sigma)$ is odd, it can be checked that $\bar{W}$ is a twin walk.
  Finally, by construction, $\bar{\ell} (\bar{W}) \leq \ell_{(C,\sigma)}(\hat{z})$ holds, where the inequality comes from the fact that the minima in the definition of $\bar{\ell}$ need not be attained by the edges from $C$.
  By definition of $s$ we have $\ell_{(C,\sigma)}(\hat{z}) = 1+s$, thus $\bar{\ell} (\bar{W}) \leq 1+s$.
  \qed
\end{proof}


\begin{mylemma}
  \label{thm_walk_to_hyperwalk}
  For each twin walk $\bar{W}$ in $\bar{G}$ there exists an odd signed closed walk $(C,\sigma)$ in $G$ whose induced simple odd $\beta$-cycle inequality~\eqref{eq_simple_odd_beta_cycle} has slack $\bar{\ell}(\bar{W}) - 1$ with respect to $\hat{z}$.
  In particular, if $\bar{\ell}(\bar{W}) < 1$ holds, then the inequality is violated by $\hat{z}$.
\end{mylemma}

\begin{proof}
  Let $\bar{W}$ be a twin walk in $\bar{G}$.
  We first construct the signed closed walk $(C,\sigma)$ by processing the edges of $\bar{W}$ in their order.
  Throughout the construction we maintain the index $i$ of the next edge to be constructed, which initially is $i \coloneqq 1$.
  Since the construction depends on the type of the current edge $\bar{e} = \{\bar{u}, \bar{v}\} \in \bar{W}$ (where $\bar{W}$ visits $\bar{u}$ first), we distinguish the relevant cases:

\medskip
\noindent
  \textbf{Case 1:} $\bar{e} \in \bar{E}^{+,+,\pm}$ and $\bar{u} \in \bar{V}_{\mathrm{E}}$.
  Hence, $\bar{u} = (e,p)$ and $\bar{v} = (v,p)$ for some $v \in e \in E$ and some $p \in \{\pm\}$.
  We define $v_i \coloneqq v$ and continue.

\medskip
\noindent
  \textbf{Case 2:} $\bar{e} \in \bar{E}^{+,+,\pm}$ and $\bar{u} \in \bar{V}_+$.  Hence, $\bar{u} = (v,p)$ and $\bar{v} = (e,p)$ for some $v \in e \in E$ and some $p \in \{\pm\}$ as well as $\ell_{\bar{e}} = \sinc{e}{v}$.
  We define $e_i \coloneqq e$ and $\sigma(i) \coloneqq +$.
  We then increase $i$ by $1$ and continue.

\medskip
\noindent
  \textbf{Case 3:} $\bar{e} \in \bar{E}^{+,-,-}$ and $\bar{u} \in V_{\mathrm{E}}$.
  Hence, $\bar{u} = (e,p)$ and $\bar{v} = (f \cap g,-p)$ for some $(e,f,g) \in \triples$ as well as $\ell_{\bar{e}} = \sone{f}{f \cap g}{e}(\hat{z})$.
  We define $v_i$ (resp.\ $v_{i+1}$) to be any node in $e \cap f$ (resp.\ $f \cap g$), $e_i \coloneqq f$ and $\sigma(i) \coloneqq -$.
  We then increase $i$ by $1$ and continue.

\medskip
\noindent
  \textbf{Case 4:} $\bar{e} \in \bar{E}^{-,-,-}$.
  Hence, $\bar{u} =(e \cap f,p)$ and $\bar{v}= (f \cap g,-p)$ for some $(e,f,g) \in \triples$ as well as $\ell_{\bar{e}} = \sodd{f}{e \cap f}{f \cap g}(\hat{z})$.
  We define $e_i \coloneqq f$, $\sigma(i) \coloneqq -$ and $v_{i+1}$ to be any node in $f \cap g$.
  We then increase $i$ by $1$ and continue.

\medskip
\noindent
  \textbf{Case 5:} $\bar{e} \in \bar{E}^{+,-,-}$ and $\bar{u} \in V_-$.
  Hence, $\bar{u} = (e \cap f,p,-)$ and $\bar{v} = (g,-p)$ for some $(e,f,g) \in \triples$ with $\ell_{\bar{e}} = \sone{f}{e \cap f}{g}(\hat{z})$.
  We define $e_i \coloneqq f$, $\sigma(i) \coloneqq -$ and $v_{i+1}$ to be any node in $f \cap g$.
  We then increase $i$ by $1$ and continue.

\medskip
\noindent
  \textbf{Case 6:} $\bar{e} \in \bar{E}^{+,-,+}$.
  Hence, $\bar{u} = (e,p)$ and $\bar{v} = (g,-p)$ for some $(e,f,g) \in \triples$ as well as $\ell_{\bar{e}} = \stwo{f}{e}{g}(\hat{z})$.
  We define $v_i$ (resp.\ $v_{i+1}$) to be any node in $e \cap f$ (resp.\ $f \cap g$), $e_i \coloneqq f$, $\sigma(i) \coloneqq -$, $e_{i+1} \coloneqq g$ and $\sigma(i+1) \coloneqq +$.
  We then increase $i$ by $2$ and continue.

\medskip
  After processing all edges of $\bar{W}$, the last defined edge is $e_{i-1}$ and thus we define $k \coloneqq i-1$ and $C \coloneqq v_1$-$e_1$-$v_2$-$e_2$-$v_3$-$\dotsb$-$v_{k-1}$-$v_{k-1}$-$v_k$-$e_k$-$v_1$.
  By checking pairs of edges of $\bar{W}$ that arise consecutively, one verifies that for each $i \in [k]$, we also have $v_i \in e_{i-1} \cap e_i$.

  To see that $(C,\sigma)$ is odd, we use the fact that the endnodes of $\bar{W}$ are twin nodes.
  When traversing an edge $\bar{e}$ from $\bar{u}$ to $\bar{v}$, the second entries of $\bar{u}$ and $\bar{v}$ differ if and only if we set a $\sigma$-entry to $-$.
  Note that in Case~6 we set two such entries, but only one to $-$.
  We conclude that $\sigma(i) = -$ holds for an odd number of indices $i \in [k]$.

  By construction we have $\bar{\ell} (\bar{W}) = \ell_{(C,\sigma)} (\hat{z})$. 
  The slack of the simple odd $\beta$-cycle inequality induced by $(C,\sigma)$ with respect to $\hat z$ is then $\ell_{(C,\sigma)}(\hat{z}) - 1 = \bar{\ell} (\bar{W}) - 1$.
  \qed
\end{proof}

\begin{theorem}
  \label{thm_simple_odd_beta_separation}
  Let $G = (V,E)$ be a hypergraph and let $\hat{z} \in \Pfr(G)$.
  The separation problem for simple odd $\beta$-cycle inequalities~\eqref{eq_simple_odd_beta_cycle} can be solved in time
  $\orderO( |E|^5 + |V|^2 \cdot |E|)$.
\end{theorem}

\begin{proof}
  Let $n \coloneqq |V|$ and $m \coloneqq |E|$ and assume $m \geq \log(n)$ since otherwise we can merge nodes that are incident to exactly the same edges.
  First note that, regarding the size of the auxiliary graph $\bar{G}$, we have $|\bar{V}| = \orderO(m^2 + n)$ and $|\bar{E}| = \orderO( mn + m^3 )$.
  For the construction of $\bar{G}$ and the computation of $\bar{\ell}$ we need to inspect all triples $(e,f,g) \in \triples{}$ of edges.
  This can be done in time $\orderO( m^3 n )$ since for each of the $m^3$ edge triples $(e,f,g)$ we have to inspect at most $n$ nodes to check the requirements on the intersections of $e$, $f$ and $g$.

  According to \cref{thm_walk_to_hyperwalk,thm_hyperwalk_to_walk} we only need to check for the existence of a twin walk $\bar{W}$ in $\bar{G}$ with $\ell(\bar{W}) < 1$.
  This can be accomplished with $|\bar{V}|/2 = \orderO( m^2 + n )$ runs of Dijkstra's algorithm~\cite{Dijkstra59} on $\bar{G}$, each of which takes
  \[
    \orderO( |\bar{E}| + |\bar{V}| \cdot \log(|\bar{V}|) ) = \orderO( (m n + m^3) + (m^2 + n) \cdot \log (m^2 + n) ) 
  \]
  time when implemented with Fibonacci heaps~\cite{FredmanT87}.
  If $m^2 \geq n$, then the total running time simplifies to $\orderO(m^5)$, and otherwise we obtain $ \orderO(n^2 m)$.
  \qed
\end{proof}

The main reason for this large running time bound is the fact that $|\bar{V}_-|$ can be quadratic in $|E|$.

Clearly, our separation algorithm requires that the edge lengths $\bar{\ell}$ of the auxiliary graph $\bar{G}$ are nonnegative.
This in turn requires $\hat{z} \in \Pfr(G)$, i.e., that the flower inequalities with at most two neighbors are satisfied.
As we already mentioned, the number of these flower inequalities is bounded by a polynomial in $|V|$ and $|E|$.
We like to point out that one can combine the separation of these flower inequalities with the construction of $\bar{G}$, i.e., one can determine violated inequalities while constructing the auxiliary graph.

\section{Redundancy}
\label{sec_redundancy}

Denote by $\Pcr(G)$ the set of points in $\Pfr(G)$ that satisfy all simple odd $\beta$-cycle inequalities.
It turns out that many such inequalities are redundant for $\Pcr(G)$.
However, the redundancy proofs do not provide much insight and often require many case distinctions (on the involved sign patterns and the way edges intersect).
Hence, we restrict ourselves to providing an example of such a redundancy result.

Since $\Pml(G)$ is full-dimensional (provided that $G$ has no loops or parallel edges) and contained in $\Pcr(G)$, then also $\Pcr(G)$ is full-dimensional.
As a consequence, an inequality is redundant for $\Pcr(G)$ if and only if it is not facet-defining for $\Pcr(G)$ (see Chapter 8 of Schrijver's book~\cite{Schrijver86}).

\begin{myproposition}
  \label{thm_redundancy_plus_plus}
  Let $(C,\sigma)$ be an odd signed closed walk in $G$ such that there exist pairwise distinct indices $i,j \in [k]$ such that $e_i = e_j$ and $\sigma(i) = \sigma(j) = +$.
  Then the simple odd $\beta$-cycle inequality corresponding to $(C,\sigma)$ is redundant for $\Pcr(G)$.
\end{myproposition}

\begin{proof}
%
%
Without loss of generality we can assume $i=1$.
Note that, from the definition of closed walk, we have
   and $4 \leq j \leq k-2$.
  We define the following two closed walks in $G$:
  $C_1 \coloneqq v_j$-$e_1$-$v_2$-$e_2$-$\dotsb$-$v_{j-1}$-$e_{j-1}$-$v_j$, and $C_2 \coloneqq v_1$-$e_j$-$v_{j+1}$-$e_{j+1}$-$\dots$-$v_k$-$e_k$-$v_1$.
  Since $4 \leq j \leq k-2$, both $C_1$ and $C_2$ consist of at least three edges.
  Let $\sigma_1$ and $\sigma_2$ be the signatures of $C_1$ and $C_2$, respectively, obtained from $\sigma$.
  Since $(C,\sigma)$ is odd, exactly one among $(C_1,\sigma_1)$ and $(C_2,\sigma_2)$ is odd, while the other is even.
  Without loss of generality we assume that $(C_1,\sigma_1)$ is odd and $(C_2,\sigma_2)$ is even.
  To prove that the simple odd $\beta$-cycle inequality $\ell_{(C,\sigma)}(z) \geq 1$ corresponding to $(C,\sigma)$ is redundant for $\Pcr(G)$, it suffices to show that it is the sum of the simple odd $\beta$-cycle inequality $\ell_{(C_1,\sigma_1)}(z) \geq 1$ corresponding to $(C_1,\sigma_1)$ and of the inequality $\ell_{(C_2,\sigma_2)}(z) \geq 0$, which is valid for $\Pfr(G)$ by \cref{thm_slack_functions}.
  To see this, it suffices to 
  consider
  the following cases, where each case not explicitly discussed below is symmetric to one of the written ones:
 
\medskip
\noindent
  \textbf{Case 1:} $1 \in I_{+,+,+}$ and $j \in I_{+,+,+}$.
  In this case we have $1 \in I_{+,+,+}$ in $C_1$ and $j \in I_{+,+,+}$ in $C_2$.
  It follows that the contribution
  $
  \big( \sinc{e_1}{v_1}(z) + \sinc{e_1}{v_2}(z) \big)
  +
  \big( \sinc{e_j}{v_j}(z) + \sinc{e_j}{v_{j+1}} \big)
  $
  from $C$ is equal to the sum of the contribution
  $
  \big( \sinc{e_1}{v_j}(z) + \sinc{e_1}{v_2}(z) \big)
  $
  from $C_1$ and the contribution 
  $
  \big( \sinc{e_j}{v_1}(z) + \sinc{e_j}{v_{j+1}}(z) \big)
  $
  from $C_2$.
  
\medskip
\noindent
  \textbf{Case 2:} $1 \in I_{+,+,+}$ and $j \in I_{+,+,-}$.
  In this case we have $1 \in I_{+,+,+}$ in $C_1$ and $j \in I_{+,+,-}$ in $C_2$.
  It follows that the contribution
  $
  \big( \sinc{e_1}{v_1}(z) + \sinc{e_1}{v_2}(z) \big)
  +
  \sinc{e_j}{v_j}(z)
  $
  from $C$ is equal to the sum of the contribution
  $
  \big(\sinc{e_1}{v_j}(z) + \sinc{e_1}{v_2}(z) \big)
  $
  from $C_1$ and the contribution 
  $
  \sinc{e_j}{v_1}(z) 
  $
  from $C_2$.

\medskip
\noindent
  \textbf{Case 3:} $1 \in I_{+,+,+}$ and $j \in I_{-,+,-}$.
  In this case we have $1 \in I_{-,+,+}$ in $C_1$ and $j \in I_{+,+,-}$ in $C_2$.
  It follows that the contribution
  $
  \big( \sinc{e_1}{v_1}(z) + \sinc{e_1}{v_2}(z) \big)
  $
  from $C$ is equal to the sum of the contribution
  $
  \sinc{e_1}{v_2}(z)
  $
  from $C_1$ and the contribution 
  $
  \sinc{e_j}{v_1}(z) 
  $
  from $C_2$.

\medskip
\noindent
  \textbf{Case 4:} $1 \in I_{+,+,-}$ and $j \in I_{+,+,-}$.
  In this case we have $1 \in I_{+,+,-}$ in $C_1$ and $j \in I_{+,+,-}$ in $C_2$.
  It follows that the contribution
  $
  \sinc{e_1}{v_1}(z) + \sinc{e_j}{v_j}(z)
  $
  from $C$ is equal to the sum of the contribution
  $
  \sinc{e_1}{v_j}(z)
  $
  from $C_1$ and the contribution 
  $
  \sinc{e_j}{v_1}(z)
  $
  from $C_2$.

\medskip
\noindent
  \textbf{Case 5:} $1 \in I_{+,+,-}$ and $j \in I_{-,+,+}$.
  In this case we have $1 \in I_{-,+,-}$ in $C_1$ and $j \in I_{+,+,+}$ in $C_2$.
  It follows that the contribution
  $
  \sinc{e_1}{v_1}(z) + \sinc{e_j}{v_{j+1}}(z)
  $
  from $C$ is equal to the sum of the contribution
  $
  0
  $
  from $C_1$ and the contribution 
  $
  \sinc{e_j}{v_1}(z) + \sinc{e_j}{v_{j+1}}(z)
  $
  from $C_2$.
  
\medskip
\noindent
  \textbf{Case 6:} $1 \in I_{+,+,-}$ and $j \in I_{-,+,-}$.
  In this case we have $1 \in I_{-,+,-}$ in $C_1$ and $j \in I_{+,+,-}$ in $C_2$.
  It follows that the contribution
  $
  \sinc{e_1}{v_1}(z)
  $
  from $C$ is equal to the sum of the contribution
  $
  0
  $
  from $C_1$ and the contribution 
  $
  \sinc{e_j}{v_1}(z)
  $
  from $C_2$.
  
\medskip
\noindent
  \textbf{Case 7:} $1 \in I_{-,+,-}$ and $j \in I_{-,+,-}$.
  In this case we have $1 \in I_{-,+,-}$ in $C_1$ and $j \in I_{-,+,-}$ in $C_2$.
  It follows that the contribution
  $
  0
  $
  from $C$ is equal to the sum of the contribution
  $
  0
  $
  from $C_1$ and the contribution 
  $
  0
  $
  from $C_2$.
\end{proof}

Note that \cref{thm_redundancy_plus_plus} is only stated for repetition of edges whose signs are both $+$.
However, we have evidence (based on computations for small instances) that also other types of closed walks yield redundant inequalities.
Examples are those with repetitions of edges of arbitrary sign, those with two subsequent equal nodes, those in which two nodes are repeated and the four involved edges all have the same sign.
The strongest redundancy statement that we can think of is captured in the following conjecture.

\begin{myconjecture}
  \label{conj_redundancy}
  Let $(C,\sigma)$ be an odd signed closed walk in $G$ for which a proper subsequence $C'$ of $C$ is a $\beta$-cycle.
  Then the simple odd $\beta$-cycle inequality corresponding to $(C,\sigma)$ is redundant for $\Pcr(G)$.
\end{myconjecture}

We recall that a \emph{$\beta$-cycle} of length $k$, for some $k \geq 3$, is a sequence $v_1$-$e_1$-$v_2$-$e_2$-$\dotsb$-$v_{k}$-$e_{k}$-$v_1$ such that $v_1$, $v_2$, $\dots$, $v_k$ are pairwise distinct nodes, $e_1$, $e_2$, $\dots$, $e_k$ are pairwise distinct edges, and $v_i$ belongs to $e_{i-1}, e_i$ and no other $e_j$ for all $i = 1,\dots,k$, where $e_0 := e_k$.
Conjecture~\ref{conj_redundancy} justifies that although inequalities~\eqref{eq_simple_odd_beta_cycle} are defined for closed walks, we call them simple odd $\beta$-\emph{cycle} inequalities.
The main reason for considering closed walks instead of $\beta$-cycles is the separation algorithm described in \cref{sec_separation}.

\section{Relation to non-simple odd $\beta$-cycle inequalities}
\label{sec_nonsimple}

In this section we relate our simple odd $\beta$-cycle inequalities 
to the odd $\beta$-cycle inequalities in~\cite{DelPiaD21}.

A \emph{cycle hypergraph} is a hypergraph $G = (V, E)$, with $E = \{e_1, \dots, e_m\}$, where $m \geq 3$, and every edge $e_i$ has nonempty intersection only with $e_{i-1}$ and $e_{i+1}$ for every $i \in \{1,\dots,m\}$, where, for convenience, we define $e_{m+1} := e_1$ and $e_0 := e_m$.
If $m = 3$, it is also required that $e_1 \cap e_2 \cap e_3 = \emptyset$. 
Given a closed walk $C = v_1$-$e_1$-$v_2$-$e_2$-$\dotsb$-$v_k$-$e_k$-$v_1$ in a hypergraph $G = (V, E)$, the \emph{support hypergraph} of $C$ is the hypergraph $G(C) = (V(C),E(C))$, where $E(C) := \{e_1,e_2,\dots,e_k\}$ and $V(C):=e_1 \cup e_2 \cup \cdots \cup e_k$.

\begin{mylemma}
  \label{thm_cycle_hypergraph_nonsimple_simple}
  Let $(C,\sigma)$ be a signed closed walk in a hypergraph $G$
  and assume that the support hypergraph of $C$ is a cycle hypergraph.
  Let $E^- \coloneqq \{e_i : i \in [k], \ \sigma(i) = -\}$, $E^+ \coloneqq \{e_i : i \in [k], \ \sigma(i) = +\}$, $S_1 \coloneqq (\bigcup_{e \in E^-} e) \setminus \bigcup_{e \in E^+} e$, and $S_2 \coloneqq \{v_1,\dots,v_k\} \setminus \bigcup_{e \in E^-} e$.
  Then
  \begin{align*}
    \ell_{(C,\sigma)}(z)
    & = - \sum_{v \in S_1} 2 z_v 
      + \sum_{e \in E^-} 2 z_e 
      + \sum_{v \in S_2} 2 z_v 
      - \sum_{e \in E^+} 2 z_e 
      + 2 |S_1| \\ 
    & \hspace{4mm} - 2 | \{ i \in [k] : e_{i-1}, e_{i} \in E^- \}| 
      + |E^-|.
  \end{align*}
  In particular, the simple odd $\beta$-cycle inequality corresponding to $(C,\sigma)$ coincides with the odd $\beta$-cycle inequality corresponding to $(C,\sigma)$.
  Furthermore, in a cycle hypergraph, every odd $\beta$-cycle inequality is a simple odd $\beta$-cycle inequality.
\end{mylemma}

\begin{proof}
  It suffices to observe that 
  \begin{align*}
    & \sum_{\substack{i \in [k] \\ \sigma(i-1) = \sigma(i) = +}} \hspace{-3ex} 2 {z}_{v_i} = \sum_{v \in S_2} 2 z_v, 
    \qquad
    \sum_{\substack{i \in [k] \\ \sigma(i-1) = \sigma(i) = - \\ v \in e_{i-1} \cap e_i}} \hspace{-3ex} 2 {z}_v
    +
    \hspace{-2ex} \sum_{\substack{i \in [k] : \sigma(i) = - \\ v \in e_i \setminus  (e_{i-1} \cup e_{i+1}) }} \hspace{-4ex} 2 {z}_v
    = \sum_{v \in S_1} 2 z_v, \\
    & \sum_{\substack{i \in [k] \\ \sigma(i-1) = \sigma(i) = - \\ v \in e_{i-1} \cap e_i}} \hspace{-1.5ex} 2
    + \hspace{-1.5ex} \sum_{\substack{i \in [k] : \sigma(i) = - \\ v \in e_i \setminus  (e_{i-1} \cup e_{i+1}) }} \hspace{-3ex} 2 
    = 2 |S_1|
    \qquad
    \text{and } 
    \sum_{\substack{i \in [k] \\ \sigma(i) = -}} 1
    = |E^-|.
  \end{align*}
  The statement for cycle hypergraphs $G$ follows by inspecting the definition of the odd $\beta$-cycle inequalities.
  \qed
\end{proof}

As a consequence, we can use the two following known results in order to gain insights about simple odd $\beta$-cycle inequalities.

\begin{myproposition}[Example~2 in~\cite{DelPiaD21}]
\label{prop ex in DD}
There exists a cycle hypergraph for which the Chv{\'a}tal rank of odd $\beta$-cycle inequalities can be equal to $2$.
\end{myproposition}

\begin{myproposition}[Implied by Theorem~1 in~\cite{DelPiaD21}]
  \label{thm_flower_cg}
  Flower inequalities are Chv{\'a}tal-Gomory cuts for $\Psr(G)$.
\end{myproposition}

\begin{theorem}
Simple odd $\beta$-cycle inequalities can have Chv{\'a}tal rank~2 with respect to $\Psr(G)$.
\end{theorem}

\begin{proof}
  Combining \cref{thm_flower_cg} with \cref{thm_simple_odd_beta_cg} shows that simple odd $\beta$-cycle inequalities have Chv{\'a}tal rank at most $2$.
  \cref{thm_cycle_hypergraph_nonsimple_simple} and \cref{prop ex in DD} show that the Chv{\'a}tal rank of simple odd $\beta$-cycle inequalities for cycle hypergraphs can be equal to 2.
  \qed
\end{proof}

For the second insight, we consider a strengthened form of Theorem~5 in~\cite{DelPiaD21}.

\begin{myproposition}[Theorem~5 in~\cite{DelPiaD21}, strengthened]
  \label{thm_cycle_hypergraph_nonsimple_perfect}
  Let $G = (V,E)$ be a cycle hypergraph.
  Then $\Pml(G)$ is described by all odd $\beta$-cycle inequalities and all inequalities from $\Pfr(G)$.
\end{myproposition}

The strengthening lies in the fact that in the original statement of Theorem~5 in~\cite{DelPiaD21} all flower inequalities are used rather than only those with at most two neighbors.
This strengthening of the original statement can be seen by inspecting its proof in~\cite{DelPiaD21}.
By applying~\cref{thm_cycle_hypergraph_nonsimple_simple} to \cref{thm_cycle_hypergraph_nonsimple_perfect} we immediately obtain the following result.

\begin{theorem}
  Let $G = (V,E)$ be a cycle hypergraph.
  Then 
  \[
    \Pml(G) = \{ x \in \Pfr(G) : x \text{ satisfies all simple odd $\beta$-cycle inequalities} \}.
  \]
\end{theorem}

\section{Implementation and computational results}
\label{sec_computations}

\DeclareDocumentCommand\ub{m}{\textbf{#1}}

We implemented the algorithm from \cref{thm_simple_odd_beta_separation} and carried out preliminary experiments in order to assess the quality of the simple odd $\beta$-cycle inequalities.
In this section we report about this work and its outcomes.

\paragraph{Implementation.}
We implemented separation algorithms for flower inequalities~\eqref{eq_flower} with at most two neighbors and for the simple odd $\beta$-cycle inequalities~\eqref{eq_simple_odd_beta_cycle} within a \texttt{separator} plugin for SCIP~\cite{SCIP8}, which is one of the state-of-the-art mixed-integer programming solvers.
SCIP automatically linearizes the objective in case it reads a binary polynomial optimization problem.
It creates so-called \texttt{and}-constraints, and hence our plugin only needs to scan all such constraints in order to build the hypergraph $G = (V,E)$.

SCIP does not add all inequalities~\eqref{eq_lin_1} to the initial LP relaxation, but starts with an aggregation of all these corresponding to each \texttt{and}-constraint.
Hence, in order to compute the dual bound $\Psr$ corresponding to $\Psr(G)$, we need to solve the separation problem for~\eqref{eq_lin_1} as well.
Moreover, we disable SCIP's default cutting planes (as well as heuristics in order to fairly compare running times).
Once no such constraint is violated, we have computed the bound from $\Psr(G)$ and start generating flower inequalities~\eqref{eq_flower} with one neighbor.
In order to avoid scanning all pairs of edges $e,f \in E$, we store, for each edge $e \in E$, its nodes $v \in e$ as well as, for each node $v \in V$, its incident edges $e \in E$.
Moreover, we keep track of all sets $e \cap f$ for $e,f \in E$ and use a hash map in order to find such sets quickly.
This is important because, for our test instances, the number of such intersection sets was not much larger than the number of edges.
We denote the obtained dual bound by $\PfrOne$.
Similarly, we obtain the dual bound $\Pfr$ by also generating violated flower inequalities~\eqref{eq_flower} with two neighbors.
We also compute the dual bound $\Pcr$ by generating violated simple odd $\beta$-cycle inequalities~\eqref{eq_simple_odd_beta_cycle}.
For each of these four bounds we continue generating cuts for the previous bounds.
In particular, we do not even solve the separation problem for flower inequalities with two neighbors or for simple odd $\beta$-cycle inequalities if inequalities of simpler types were found.
Finally, we also run SCIP with all default separators enabled and denote the corresponding bound by $\boundSCIP$.

Since for Dijkstra's algorithm, Fibonacci heaps outperform min-heaps only for very large graphs, and since SCIP already contains a min-heap implementation, we sticked to the latter.
Technically, this yields a higher amortized running time of $\orderO((|V| \cdot |E| + |E|^3) \log(|V| + |E|^2))$ for each call of Dijkstra's algorithm, and thus a total running time of $\orderO( |E|^5 \cdot \log |V| + |V|^2 \cdot |E| \cdot \log |V| )$.
While this is a logarithmic factor worse than  the running time from \cref{thm_simple_odd_beta_separation}, we believe that switching to Fibonacci heaps would not improve the practical performance.

\paragraph{Instances.}
The first set of instances of binary polynomial optimization that we considered, is inspired from the \emph{image restoration problem}, which is widely investigated in computer vision.
The goal is to reconstruct an original base image from a blurred image in input.
These instances have been considered also in \cite{Rodriguez-Heck18,ElloumiLL21,DelPiaKS20,DelPiaD21,DelPiaK21} to test binary polynomial optimization algorithms.
A complete description of these instances can be found in \cite{CramaR17}, while here we only provide a brief overview.
An image is a rectangle containing $\ell \times h$ pixels, and each pixel takes the value $0$ or $1$.
Three types of base images are considered: 
The first type ``topleft'' consists of a rectangle of $1$s in the top left, the second type ``center'' contains a rectangle of $1$s in the middle, and in the third type ``cross'' the $1$s form a cross in the center.
Three different types of perturbations can be applied to a base image to obtain a blurred image: ``none'', where no perturbation is employed, ``all \SI{5}{\percent}'', where each pixel is flipped independently with probability \SI{5}{\percent}, and ``0s \SI{50}{\percent}'', where each $0$-pixel is flipped independently with probability \SI{50}{\percent} (and the $1$s remain unchanged).
The image restoration instance associated with a blurred image is defined by an objective function $f(x) = L(x) + P(x)$ that must be minimized. 
The variables $x_{ij}$, for all $(i, j) \in [\ell] \times [h]$, represent the value assigned to each pixel in the output image.
$L(x)$ is the linear part which incentivizes similarity between the input blurred image and the output image.
$P(x)$ is the nonlinear part, consisting of a polynomial of degree four, which incentivizes smoothness by penalizing $2 \times 2$ sub-images of the output image, the more they look like a checkerboard.
We considered seven different sizes $(\ell,h) \in \{(10,10), (10,15), (15,15), (15,20), (20,20), (20,25), (25,25) \}$.
For each size, we considered each of the three different base types and each of the three perturbations, where for ``all \SI{5}{\percent}'' and ``0s \SI{50}{\percent}'' we created two random instances, resulting in a total of $3 \cdot 5 = 15$ instances for each size.
This is in line with~\cite{CramaR17} where these instances were considered for the first time.

The second set of instances that we considered, arises from the \emph{low auto-correlation binary sequence problem}, which originates in theoretical physics \cite{LiersMPRS10}.
Both POLIP \cite{POLIP14} and MINLPLib~\cite{MINLPLib20} contain 44 instances of this type, which have been already used in the context of binary polynomial optimization in \cite{Rodriguez-Heck18,ElloumiLL21,DelPiaD21}.
These instances are notoriously very difficult, and exhibit several symmetries, which in particular lead to multiple optimal solutions.
%
Determining a ground state in the so-called Bernasconi model amounts to minimizing the degree-four energy function 
\begin{equation}
  \label{eq energy}
  \sum_{i=1}^{N-R+1} \ 
  \sum_{d=1}^{R-1} \ 
  \left(
  \sum_{j=1}^{i+R-1-d} \ 
  \sigma_j \sigma_{j+d} 
  \right) ^2,
\end{equation}
over variables $\sigma_j$, for $j=1,\dots,N$, taking values in $\{+1,-1\}$.
In the formula \eqref{eq energy}, $N$ and $R$ are positive constants, and the formula is not equal to a constant for $R \in \{3,\dots,N\}$.
Note that \eqref{eq energy} coincides with the energy function (2) in \cite{LiersMPRS10}, up to multiplication by a positive constant, and renaming of variables so that the first variable is $\sigma_1$ rather than $\sigma_0$.
To express the energy function \eqref{eq energy} in 0/1 variables, we replace $\sigma_i$ with $2 x_i-1$, for $i = 1,\dots,N$, and obtain the binary polynomial optimization problem
\begin{subequations}
  \label{eq energy bpo}
  \begin{alignat}{7}
    & \text{min } & \sum_{i=1}^{N-R+1} \sum_{d=1}^{R-1} & \left( \sum_{j=1}^{i+R-1-d} (2 x_j-1) (2x_{j+d}-1) \right)^2 \\
    & \text{s.t. } & x \in \{0,1\}^N &.
  \end{alignat}
\end{subequations}
The instances in POLIP and MINLPLib ignore the absolute term in the objective function, which makes it hard to compare bounds to those in the physics literature.
Moreover, the instance for $N = R = 20$ is missing in both benchmark libraries.
Hence, we generated our instances according to \eqref{eq energy bpo} by setting $N$ and $R$ to $N \in \{20,25,30,\dotsc,55\}$ and $R \in \{\lceil \frac{1}{8}N \rceil, \lceil \frac{2}{8}N \rceil, \lceil \frac{3}{8}N \rceil, \dotsc, \lceil \frac{7}{8} N \rceil, N\}$, respectively.

The third set of instances that we considered are generated randomly, as it is commonly done in the literature \cite{BucRin07,CramaR17,DelPiaKS20,dPDiG22SODA,dPDiG23ALG}.
We chose a setting similar to the one used in~\cite{CramaR17,dPDiG22SODA,dPDiG23ALG}, where we fix the number of nodes $|V|$ and of edges $|E|$ of the hypergraph representing the instance, and we also fix the minimum cardinality of an edge of the hypergraph, which we denote by $d$. 
For every edge, the probability that its cardinality is equal to $c$, for $c \in \{d,\dots,|V|\}$, is equal to $2^{d-c-1}/(1-2^{d-|V|-1})$, so it is roughly $\tfrac{1}{2}$ for $c=d$, $\tfrac{1}{4}$ for $c=d+1$, $\tfrac{1}{8}$ for $c=d+2$, etc..
As explained in \cite{CramaR17}, the purpose of this choice is to model the fact that a random hypergraph is expected to have more edges of low cardinality than high cardinality.
Then, once $c$ is fixed, the nodes of the edge are chosen independently uniformly at random in $V$ with no repetitions.
We also make sure that there are no parallel edges in the produced hypergraph.
We generated instances with $|V| \in \{ 50, 100, 200 \}$, $|E| \in \{5 \cdot |V|, 10 \cdot |V|, 20 \cdot |V|\}$, and $d \in \{2,3,4\}$.
For every triple $(|V|,|E|,d)$ we generated $5$ instances.
We remark that the normalization $1/(1-2^{d-|V|-1})$ in our probabilities is not present in~\cite{CramaR17,dPDiG22SODA,dPDiG23ALG}, but is added here because the cardinality $c$ of the edge can be at most $|V|$.
Another difference with the setting examined in~\cite{CramaR17,dPDiG22SODA,dPDiG23ALG} is that in these papers only the case $d=2$ is considered.

\DeclareDocumentCommand\timeout{}{$>$\,\SI{1}{\hour}}

\paragraph{Experiments.}
The main goal of our computations is of preliminary nature: we wanted to assess whether the generation of simple odd $\beta$-cycle inequalities has the potential of being useful at all.
Obviously, even a very careful implementation suffers from the large asymptotic running time of our separation algorithm.
However, the actual additional gap that can be closed by adding the inequalities was hard to predict.
In fact, our expectations were quite low because it is well known for stable set and maximum cut problems that one can often add a huge number of odd cycle inequalities with only moderate change of the dual bound~\cite{RebennackOTSRP11}.
Hence, the purpose of our experiments was to investigate the impact of simple odd $\beta$-cycle inequalities on the dual bound.

To this end, we computed the dual bounds $\Psr$, $\PfrOne$, $\Pfr$, $\Pcr$ and $\boundSCIP$.
We evaluated the quality of the bounds and the running times within the root node.
In particular, we turned off other features such as presolving, symmetry detection and heuristics, and only solved the root node.
We allowed as many separation rounds as necessary until no more inequalities could be added.

To obtain a reference objective value we first ran SCIP, augmented with flower inequalities of up to one neighbor, with a time limit of \SI{1}{\hour} and extracted the primal bound $\textup{PRIMAL}$.

We used SCIP~8.0~\cite{SCIP8} and ran our experiments on a Intel Xeon Gold 5217 CPU with \SI{3.00}{\giga\hertz} with \SI{64}{\giga\byte} memory.

\paragraph{Results.}
Next, we report our computational results.
In all the tables that we present, we provide integrality gaps and computation times for pure cutting plane procedures for relaxations $\Psr$, $\PfrOne$, $\Pfr$, $\Pcr$ (all without other cutting planes), as well as SCIP with its default cutting planes.
The gap is defined as $(\textup{db} - \Psr) / (\textup{PRIMAL} - \Psr)$ in percent, where $\textup{db}$ denotes the corresponding dual bound.

We first report about results for the image restoration instances which forms the easier instance set.
In \cref{tab_image_restoration_topleft,tab_image_restoration_center,tab_image_restoration_cross} we depict the results for these instances.
First of all, the results for the three types are all very similar.
While the computation times slowly grow for larger instances, the closed gaps are almost independent of the instance parameters.
The most important observation is that $\PfrOne = \Pfr = \Pcr$ holds throughout.
In other words neither flower inequalities with two neighbors nor simple odd $\beta$-cycle inequalities were generated.
The amount of gap closed by SCIP's default cutting planes is also similar to that of $\PfrOne$, but the separation algorithms require more time.


\begin{table}[htpb]
  \caption{%
    Results for image restoration problems with base image of type ``topleft''.
  }
  \label{tab_image_restoration_topleft}
  {\footnotesize{%
  \begin{center}
    \setlength{\tabcolsep}{1.5mm}
    \begin{tabular}{rr|rr|r|rr|rr|rrr|rr}
      \textbf{Size} & \textbf{Pert.} & $|V|$ & $|E|$ & $\Psr$ & \multicolumn{2}{c|}{\textbf{$\PfrOne$}} & \multicolumn{2}{c|}{\textbf{$\Pfr$}} & \multicolumn{3}{c|}{\textbf{$\Pcr$}} & \multicolumn{2}{c}{SCIP} \\
      & & & & \textbf{time} & \textbf{gap} & \textbf{time} & \textbf{gap} & \textbf{time} & \textbf{gap} & \textbf{time} & \textbf{cycles} & \textbf{gap} & \textbf{time} \\ \hline
 $10 \times 10$ & none & 100 & 567 & \SI{0}{\second} & \SI{49}{\percent} & \SI{0}{\second} & \SI{49}{\percent} & \SI{0}{\second} & \SI{49}{\percent} & \SI{0}{\second} & \num{0} & \SI{53}{\percent} & \SI{5}{\second}\\
 $10 \times 10$ & all 5\% & 100 & 567 & \SI{0}{\second} & \SI{48}{\percent} & \SI{0}{\second} & \SI{48}{\percent} & \SI{0}{\second} & \SI{48}{\percent} & \SI{0}{\second} & \num{0} & \SI{54}{\percent} & \SI{6}{\second}\\
 $10 \times 10$ & 0s 50\% & 100 & 567 & \SI{0}{\second} & \SI{48}{\percent} & \SI{0}{\second} & \SI{48}{\percent} & \SI{0}{\second} & \SI{48}{\percent} & \SI{0}{\second} & \num{0} & \SI{49}{\percent} & \SI{7}{\second}\\
 $10 \times 15$ & none & 150 & 882 & \SI{0}{\second} & \SI{49}{\percent} & \SI{0}{\second} & \SI{49}{\percent} & \SI{0}{\second} & \SI{49}{\percent} & \SI{0}{\second} & \num{0} & \SI{51}{\percent} & \SI{6}{\second}\\
 $10 \times 15$ & all 5\% & 150 & 882 & \SI{0}{\second} & \SI{48}{\percent} & \SI{0}{\second} & \SI{48}{\percent} & \SI{0}{\second} & \SI{48}{\percent} & \SI{0}{\second} & \num{0} & \SI{48}{\percent} & \SI{4}{\second}\\
 $10 \times 15$ & 0s 50\% & 150 & 882 & \SI{0}{\second} & \SI{46}{\percent} & \SI{0}{\second} & \SI{46}{\percent} & \SI{0}{\second} & \SI{46}{\percent} & \SI{0}{\second} & \num{0} & \SI{45}{\percent} & \SI{9}{\second}\\
 $15 \times 15$ & none & 225 & 1372 & \SI{0}{\second} & \SI{49}{\percent} & \SI{1}{\second} & \SI{49}{\percent} & \SI{0}{\second} & \SI{49}{\percent} & \SI{1}{\second} & \num{0} & \SI{47}{\percent} & \SI{4}{\second}\\
 $15 \times 15$ & all 5\% & 225 & 1372 & \SI{0}{\second} & \SI{47}{\percent} & \SI{0}{\second} & \SI{47}{\percent} & \SI{1}{\second} & \SI{47}{\percent} & \SI{1}{\second} & \num{0} & \SI{44}{\percent} & \SI{4}{\second}\\
 $15 \times 15$ & 0s 50\% & 225 & 1372 & \SI{0}{\second} & \SI{46}{\percent} & \SI{1}{\second} & \SI{46}{\percent} & \SI{1}{\second} & \SI{46}{\percent} & \SI{1}{\second} & \num{0} & \SI{43}{\percent} & \SI{7}{\second}\\
 $15 \times 20$ & none & 300 & 1862 & \SI{0}{\second} & \SI{49}{\percent} & \SI{1}{\second} & \SI{49}{\percent} & \SI{1}{\second} & \SI{49}{\percent} & \SI{1}{\second} & \num{0} & \SI{44}{\percent} & \SI{12}{\second}\\
 $15 \times 20$ & all 5\% & 300 & 1862 & \SI{0}{\second} & \SI{48}{\percent} & \SI{1}{\second} & \SI{48}{\percent} & \SI{1}{\second} & \SI{48}{\percent} & \SI{1}{\second} & \num{0} & \SI{43}{\percent} & \SI{10}{\second}\\
 $15 \times 20$ & 0s 50\% & 300 & 1862 & \SI{0}{\second} & \SI{45}{\percent} & \SI{1}{\second} & \SI{45}{\percent} & \SI{1}{\second} & \SI{45}{\percent} & \SI{1}{\second} & \num{0} & \SI{39}{\percent} & \SI{8}{\second}\\
 $20 \times 20$ & none & 400 & 2527 & \SI{1}{\second} & \SI{49}{\percent} & \SI{2}{\second} & \SI{49}{\percent} & \SI{2}{\second} & \SI{49}{\percent} & \SI{2}{\second} & \num{0} & \SI{41}{\percent} & \SI{10}{\second}\\
 $20 \times 20$ & all 5\% & 400 & 2527 & \SI{1}{\second} & \SI{48}{\percent} & \SI{2}{\second} & \SI{48}{\percent} & \SI{2}{\second} & \SI{48}{\percent} & \SI{2}{\second} & \num{0} & \SI{39}{\percent} & \SI{16}{\second}\\
 $20 \times 20$ & 0s 50\% & 400 & 2527 & \SI{1}{\second} & \SI{45}{\percent} & \SI{2}{\second} & \SI{45}{\percent} & \SI{2}{\second} & \SI{45}{\percent} & \SI{2}{\second} & \num{0} & \SI{37}{\percent} & \SI{15}{\second}\\
 $20 \times 25$ & none & 500 & 3192 & \SI{2}{\second} & \SI{49}{\percent} & \SI{2}{\second} & \SI{49}{\percent} & \SI{3}{\second} & \SI{49}{\percent} & \SI{2}{\second} & \num{0} & \SI{44}{\percent} & \SI{14}{\second}\\
 $20 \times 25$ & all 5\% & 500 & 3192 & \SI{2}{\second} & \SI{47}{\percent} & \SI{2}{\second} & \SI{47}{\percent} & \SI{2}{\second} & \SI{47}{\percent} & \SI{3}{\second} & \num{0} & \SI{38}{\percent} & \SI{14}{\second}\\
 $20 \times 25$ & 0s 50\% & 500 & 3192 & \SI{1}{\second} & \SI{45}{\percent} & \SI{3}{\second} & \SI{45}{\percent} & \SI{3}{\second} & \SI{45}{\percent} & \SI{3}{\second} & \num{0} & \SI{35}{\percent} & \SI{22}{\second}\\
 $25 \times 25$ & none & 625 & 4032 & \SI{2}{\second} & \SI{49}{\percent} & \SI{2}{\second} & \SI{49}{\percent} & \SI{2}{\second} & \SI{49}{\percent} & \SI{3}{\second} & \num{0} & \SI{41}{\percent} & \SI{17}{\second}\\
 $25 \times 25$ & all 5\% & 625 & 4032 & \SI{2}{\second} & \SI{47}{\percent} & \SI{3}{\second} & \SI{47}{\percent} & \SI{3}{\second} & \SI{47}{\percent} & \SI{4}{\second} & \num{0} & \SI{36}{\percent} & \SI{24}{\second}\\
 $25 \times 25$ & 0s 50\% & 625 & 4032 & \SI{2}{\second} & \SI{44}{\percent} & \SI{4}{\second} & \SI{44}{\percent} & \SI{5}{\second} & \SI{44}{\percent} & \SI{5}{\second} & \num{0} & \SI{32}{\percent} & \SI{37}{\second}
    \end{tabular}
  \end{center}
  }}%
\end{table}


\begin{table}[htpb]
  \caption{%
    Results for image restoration problems with base image of type ``center''.
  }
  \label{tab_image_restoration_center}
  {\footnotesize{%
  \begin{center}
    \setlength{\tabcolsep}{1.5mm}
    \begin{tabular}{rr|rr|r|rr|rr|rrr|rr}
      \textbf{Size} & \textbf{Pert.} & $|V|$ & $|E|$ & $\Psr$ & \multicolumn{2}{c|}{\textbf{$\PfrOne$}} & \multicolumn{2}{c|}{\textbf{$\Pfr$}} & \multicolumn{3}{c|}{\textbf{$\Pcr$}} & \multicolumn{2}{c}{SCIP} \\
      & & & & \textbf{time} & \textbf{gap} & \textbf{time} & \textbf{gap} & \textbf{time} & \textbf{gap} & \textbf{time} & \textbf{cycles} & \textbf{gap} & \textbf{time} \\ \hline
 $10 \times 10$ & none & 100 & 567 & \SI{0}{\second} & \SI{47}{\percent} & \SI{0}{\second} & \SI{47}{\percent} & \SI{0}{\second} & \SI{47}{\percent} & \SI{0}{\second} & \num{0} & \SI{53}{\percent} & \SI{7}{\second}\\
 $10 \times 10$ & all 5\% & 100 & 567 & \SI{0}{\second} & \SI{46}{\percent} & \SI{0}{\second} & \SI{46}{\percent} & \SI{0}{\second} & \SI{46}{\percent} & \SI{0}{\second} & \num{0} & \SI{51}{\percent} & \SI{8}{\second}\\
 $10 \times 10$ & 0s 50\% & 100 & 567 & \SI{0}{\second} & \SI{52}{\percent} & \SI{0}{\second} & \SI{52}{\percent} & \SI{0}{\second} & \SI{52}{\percent} & \SI{0}{\second} & \num{0} & \SI{56}{\percent} & \SI{6}{\second}\\
 $10 \times 15$ & none & 150 & 882 & \SI{0}{\second} & \SI{48}{\percent} & \SI{0}{\second} & \SI{48}{\percent} & \SI{0}{\second} & \SI{48}{\percent} & \SI{0}{\second} & \num{0} & \SI{52}{\percent} & \SI{4}{\second}\\
 $10 \times 15$ & all 5\% & 150 & 882 & \SI{0}{\second} & \SI{47}{\percent} & \SI{0}{\second} & \SI{47}{\percent} & \SI{0}{\second} & \SI{47}{\percent} & \SI{0}{\second} & \num{0} & \SI{51}{\percent} & \SI{4}{\second}\\
 $10 \times 15$ & 0s 50\% & 150 & 882 & \SI{0}{\second} & \SI{47}{\percent} & \SI{0}{\second} & \SI{47}{\percent} & \SI{0}{\second} & \SI{47}{\percent} & \SI{0}{\second} & \num{0} & \SI{50}{\percent} & \SI{4}{\second}\\
 $15 \times 15$ & none & 225 & 1372 & \SI{0}{\second} & \SI{48}{\percent} & \SI{0}{\second} & \SI{48}{\percent} & \SI{0}{\second} & \SI{48}{\percent} & \SI{1}{\second} & \num{0} & \SI{53}{\percent} & \SI{8}{\second}\\
 $15 \times 15$ & all 5\% & 225 & 1372 & \SI{0}{\second} & \SI{47}{\percent} & \SI{1}{\second} & \SI{47}{\percent} & \SI{1}{\second} & \SI{47}{\percent} & \SI{1}{\second} & \num{0} & \SI{47}{\percent} & \SI{14}{\second}\\
 $15 \times 15$ & 0s 50\% & 225 & 1372 & \SI{0}{\second} & \SI{48}{\percent} & \SI{1}{\second} & \SI{48}{\percent} & \SI{1}{\second} & \SI{48}{\percent} & \SI{1}{\second} & \num{0} & \SI{43}{\percent} & \SI{6}{\second}\\
 $15 \times 20$ & none & 300 & 1862 & \SI{0}{\second} & \SI{48}{\percent} & \SI{1}{\second} & \SI{48}{\percent} & \SI{1}{\second} & \SI{48}{\percent} & \SI{1}{\second} & \num{0} & \SI{44}{\percent} & \SI{6}{\second}\\
 $15 \times 20$ & all 5\% & 300 & 1862 & \SI{0}{\second} & \SI{47}{\percent} & \SI{1}{\second} & \SI{47}{\percent} & \SI{1}{\second} & \SI{47}{\percent} & \SI{1}{\second} & \num{0} & \SI{44}{\percent} & \SI{14}{\second}\\
 $15 \times 20$ & 0s 50\% & 300 & 1862 & \SI{1}{\second} & \SI{46}{\percent} & \SI{1}{\second} & \SI{46}{\percent} & \SI{1}{\second} & \SI{46}{\percent} & \SI{1}{\second} & \num{0} & \SI{39}{\percent} & \SI{9}{\second}\\
 $20 \times 20$ & none & 400 & 2527 & \SI{1}{\second} & \SI{48}{\percent} & \SI{2}{\second} & \SI{48}{\percent} & \SI{2}{\second} & \SI{48}{\percent} & \SI{2}{\second} & \num{0} & \SI{41}{\percent} & \SI{11}{\second}\\
 $20 \times 20$ & all 5\% & 400 & 2527 & \SI{1}{\second} & \SI{46}{\percent} & \SI{2}{\second} & \SI{46}{\percent} & \SI{2}{\second} & \SI{46}{\percent} & \SI{2}{\second} & \num{0} & \SI{39}{\percent} & \SI{9}{\second}\\
 $20 \times 20$ & 0s 50\% & 400 & 2527 & \SI{1}{\second} & \SI{45}{\percent} & \SI{2}{\second} & \SI{45}{\percent} & \SI{2}{\second} & \SI{45}{\percent} & \SI{2}{\second} & \num{0} & \SI{34}{\percent} & \SI{16}{\second}\\
 $20 \times 25$ & none & 500 & 3192 & \SI{2}{\second} & \SI{48}{\percent} & \SI{3}{\second} & \SI{48}{\percent} & \SI{3}{\second} & \SI{48}{\percent} & \SI{3}{\second} & \num{0} & \SI{40}{\percent} & \SI{17}{\second}\\
 $20 \times 25$ & all 5\% & 500 & 3192 & \SI{1}{\second} & \SI{47}{\percent} & \SI{3}{\second} & \SI{47}{\percent} & \SI{3}{\second} & \SI{47}{\percent} & \SI{3}{\second} & \num{0} & \SI{39}{\percent} & \SI{13}{\second}\\
 $20 \times 25$ & 0s 50\% & 500 & 3192 & \SI{2}{\second} & \SI{45}{\percent} & \SI{3}{\second} & \SI{45}{\percent} & \SI{3}{\second} & \SI{45}{\percent} & \SI{3}{\second} & \num{0} & \SI{33}{\percent} & \SI{16}{\second}\\
 $25 \times 25$ & none & 625 & 4032 & \SI{2}{\second} & \SI{48}{\percent} & \SI{3}{\second} & \SI{48}{\percent} & \SI{3}{\second} & \SI{48}{\percent} & \SI{5}{\second} & \num{0} & \SI{40}{\percent} & \SI{26}{\second}\\
 $25 \times 25$ & all 5\% & 625 & 4032 & \SI{2}{\second} & \SI{47}{\percent} & \SI{3}{\second} & \SI{47}{\percent} & \SI{3}{\second} & \SI{47}{\percent} & \SI{4}{\second} & \num{0} & \SI{40}{\percent} & \SI{30}{\second}\\
 $25 \times 25$ & 0s 50\% & 625 & 4032 & \SI{3}{\second} & \SI{45}{\percent} & \SI{4}{\second} & \SI{45}{\percent} & \SI{4}{\second} & \SI{45}{\percent} & \SI{5}{\second} & \num{0} & \SI{33}{\percent} & \SI{33}{\second}
    \end{tabular}
  \end{center}
  }}%
\end{table}


\begin{table}[htpb]
  \caption{%
    Results for image restoration problems with base image of type ``cross''.
  }
  \label{tab_image_restoration_cross}
  {\footnotesize{%
  \begin{center}
    \setlength{\tabcolsep}{1.5mm}
    \begin{tabular}{rr|rr|r|rr|rr|rrr|rr}
      \textbf{Size} & \textbf{Pert.} & $|V|$ & $|E|$ & $\Psr$ & \multicolumn{2}{c|}{\textbf{$\PfrOne$}} & \multicolumn{2}{c|}{\textbf{$\Pfr$}} & \multicolumn{3}{c|}{\textbf{$\Pcr$}} & \multicolumn{2}{c}{SCIP} \\
      & & & & \textbf{time} & \textbf{gap} & \textbf{time} & \textbf{gap} & \textbf{time} & \textbf{gap} & \textbf{time} & \textbf{cycles} & \textbf{gap} & \textbf{time} \\ \hline
 $10 \times 10$ & none & 100 & 567 & \SI{0}{\second} & \SI{45}{\percent} & \SI{0}{\second} & \SI{45}{\percent} & \SI{0}{\second} & \SI{45}{\percent} & \SI{0}{\second} & \num{0} & \SI{52}{\percent} & \SI{5}{\second}\\
 $10 \times 10$ & all 5\% & 100 & 567 & \SI{0}{\second} & \SI{45}{\percent} & \SI{0}{\second} & \SI{45}{\percent} & \SI{0}{\second} & \SI{45}{\percent} & \SI{0}{\second} & \num{0} & \SI{46}{\percent} & \SI{5}{\second}\\
 $10 \times 10$ & 0s 50\% & 100 & 567 & \SI{0}{\second} & \SI{47}{\percent} & \SI{0}{\second} & \SI{47}{\percent} & \SI{0}{\second} & \SI{47}{\percent} & \SI{0}{\second} & \num{0} & \SI{48}{\percent} & \SI{7}{\second}\\
 $10 \times 15$ & none & 150 & 882 & \SI{0}{\second} & \SI{45}{\percent} & \SI{0}{\second} & \SI{45}{\percent} & \SI{0}{\second} & \SI{45}{\percent} & \SI{0}{\second} & \num{0} & \SI{43}{\percent} & \SI{4}{\second}\\
 $10 \times 15$ & all 5\% & 150 & 882 & \SI{0}{\second} & \SI{45}{\percent} & \SI{0}{\second} & \SI{45}{\percent} & \SI{0}{\second} & \SI{45}{\percent} & \SI{0}{\second} & \num{0} & \SI{46}{\percent} & \SI{8}{\second}\\
 $10 \times 15$ & 0s 50\% & 150 & 882 & \SI{0}{\second} & \SI{48}{\percent} & \SI{0}{\second} & \SI{48}{\percent} & \SI{0}{\second} & \SI{48}{\percent} & \SI{0}{\second} & \num{0} & \SI{46}{\percent} & \SI{5}{\second}\\
 $15 \times 15$ & none & 225 & 1372 & \SI{0}{\second} & \SI{46}{\percent} & \SI{1}{\second} & \SI{46}{\percent} & \SI{1}{\second} & \SI{46}{\percent} & \SI{1}{\second} & \num{0} & \SI{46}{\percent} & \SI{6}{\second}\\
 $15 \times 15$ & all 5\% & 225 & 1372 & \SI{0}{\second} & \SI{44}{\percent} & \SI{0}{\second} & \SI{44}{\percent} & \SI{0}{\second} & \SI{44}{\percent} & \SI{1}{\second} & \num{0} & \SI{42}{\percent} & \SI{7}{\second}\\
 $15 \times 15$ & 0s 50\% & 225 & 1372 & \SI{0}{\second} & \SI{45}{\percent} & \SI{1}{\second} & \SI{45}{\percent} & \SI{1}{\second} & \SI{45}{\percent} & \SI{1}{\second} & \num{0} & \SI{40}{\percent} & \SI{8}{\second}\\
 $15 \times 20$ & none & 300 & 1862 & \SI{0}{\second} & \SI{46}{\percent} & \SI{1}{\second} & \SI{46}{\percent} & \SI{1}{\second} & \SI{46}{\percent} & \SI{1}{\second} & \num{0} & \SI{40}{\percent} & \SI{13}{\second}\\
 $15 \times 20$ & all 5\% & 300 & 1862 & \SI{0}{\second} & \SI{45}{\percent} & \SI{1}{\second} & \SI{45}{\percent} & \SI{1}{\second} & \SI{45}{\percent} & \SI{1}{\second} & \num{0} & \SI{44}{\percent} & \SI{9}{\second}\\
 $15 \times 20$ & 0s 50\% & 300 & 1862 & \SI{0}{\second} & \SI{44}{\percent} & \SI{1}{\second} & \SI{44}{\percent} & \SI{1}{\second} & \SI{44}{\percent} & \SI{1}{\second} & \num{0} & \SI{38}{\percent} & \SI{12}{\second}\\
 $20 \times 20$ & none & 400 & 2527 & \SI{1}{\second} & \SI{46}{\percent} & \SI{2}{\second} & \SI{46}{\percent} & \SI{2}{\second} & \SI{46}{\percent} & \SI{2}{\second} & \num{0} & \SI{39}{\percent} & \SI{8}{\second}\\
 $20 \times 20$ & all 5\% & 400 & 2527 & \SI{1}{\second} & \SI{45}{\percent} & \SI{2}{\second} & \SI{45}{\percent} & \SI{2}{\second} & \SI{45}{\percent} & \SI{2}{\second} & \num{0} & \SI{39}{\percent} & \SI{13}{\second}\\
 $20 \times 20$ & 0s 50\% & 400 & 2527 & \SI{1}{\second} & \SI{44}{\percent} & \SI{2}{\second} & \SI{44}{\percent} & \SI{2}{\second} & \SI{44}{\percent} & \SI{2}{\second} & \num{0} & \SI{35}{\percent} & \SI{17}{\second}\\
 $20 \times 25$ & none & 500 & 3192 & \SI{1}{\second} & \SI{46}{\percent} & \SI{2}{\second} & \SI{46}{\percent} & \SI{2}{\second} & \SI{46}{\percent} & \SI{3}{\second} & \num{0} & \SI{37}{\percent} & \SI{19}{\second}\\
 $20 \times 25$ & all 5\% & 500 & 3192 & \SI{1}{\second} & \SI{45}{\percent} & \SI{2}{\second} & \SI{45}{\percent} & \SI{2}{\second} & \SI{45}{\percent} & \SI{2}{\second} & \num{0} & \SI{35}{\percent} & \SI{17}{\second}\\
 $20 \times 25$ & 0s 50\% & 500 & 3192 & \SI{1}{\second} & \SI{43}{\percent} & \SI{2}{\second} & \SI{43}{\percent} & \SI{2}{\second} & \SI{43}{\percent} & \SI{3}{\second} & \num{0} & \SI{34}{\percent} & \SI{23}{\second}\\
 $25 \times 25$ & none & 625 & 4032 & \SI{2}{\second} & \SI{47}{\percent} & \SI{3}{\second} & \SI{47}{\percent} & \SI{3}{\second} & \SI{47}{\percent} & \SI{4}{\second} & \num{0} & \SI{36}{\percent} & \SI{21}{\second}\\
 $25 \times 25$ & all 5\% & 625 & 4032 & \SI{2}{\second} & \SI{45}{\percent} & \SI{4}{\second} & \SI{45}{\percent} & \SI{4}{\second} & \SI{45}{\percent} & \SI{4}{\second} & \num{0} & \SI{36}{\percent} & \SI{25}{\second}\\
 $25 \times 25$ & 0s 50\% & 625 & 4032 & \SI{2}{\second} & \SI{43}{\percent} & \SI{4}{\second} & \SI{43}{\percent} & \SI{4}{\second} & \SI{43}{\percent} & \SI{5}{\second} & \num{0} & \SI{31}{\percent} & \SI{37}{\second}
    \end{tabular}
  \end{center}
  }}%
\end{table}


Let us now turn to the results for low auto-correlation binary sequence instances.
These are much harder to solve, in particular because the hypergraphs are much more dense, which leads to large numbers of linearization variables.
In \cref{tab_auto_correlation1,tab_auto_correlation2,tab_auto_correlation3} we depict the results for these instances.
SCIP's general purpose cutting planes perform much worse than the flower inequalities with one neighbor.
Similar to the image restoration instances, no flower inequalities with two neighbors were generated.
In contrast to the previous instance class, our proposed simple odd $\beta$-cycle inequalities close gap in addition to the flower inequalities.
However, the running times of our separation algorithm are extremely large.
In fact, for a number of instances we cannot even compute the dual bound within one hour.
One reason for this is clearly the density of the hypergraph, which results in the large number of hyperedges.

Results for the random high-degree instances are depicted in \cref{tab_high_degree}.
It turns out that the addition of flower inequalities with two neighbors ($\Pfr$) on top of those with only one neighbor ($\PfrOne$) closes the more gap the larger the minimum degree $d$.
Moreover, in most cases the addition of simple odd $\beta$-cycle inequalities closes even more gap.
This strengthening is quite significant even for large instance sizes.
However, this comes at the cost of generating a large number of such inequalities.
For instance, for the last instance with $(n,m,d) = (200,4000,4)$ our algorithm generated (on average) \num{5466} inequalities~\eqref{eq_lin_1}, \num{7820.6} and \num{4212.6} flower inequalities (with one and two neighbors, respectively) and \num{46140.6} simple odd $\beta$-cycle inequalities, that is, a total of \num{63639.8} inequalities were generated to compute $\Pcr$ (with a gap of \SI{31}{\percent}).
This is in contrast to the computation of $\Pfr$ (with a gap of \SI{27}{\percent}) for which (on average) $\num{3599.4} + \num{7358.4} + \num{2431.8} = \num{13389.6}$ inequalities were required.
We conclude that for this problem class the addition of simple odd $\beta$-cycle inequalities has some potential, but it is unclear how to deal with the huge number of inequalities (besides its computation time).

The running time bound $\orderO(|E|^5 + |V|^2 \cdot |E|)$ in \cref{thm_simple_odd_beta_separation} and that of the implemented algorithm $\orderO(|E|^5 \cdot \log |V| + |V|^2 \cdot |E| \cdot \log |V|)$ is a worst-case bound.
The natural question is of whether this asymptotic behavior is also realized for our instances.
To this end, we considered all minimum-weight odd cycle computations that were carried out for the computation of $\Pcr$ for all random high degree instances.
Let us assume that the running time can be approximated by $C \cdot |E|^\alpha$ for constants $C$ and $\alpha$.
For each such computation we can compare the measured time to the predicted one and obtain the multiplicative error.
The ratio between the largest and the smallest such multiplicative errors is independent of $C$.
\Cref{fig_moc} shows how these errors depend on the estimated exponent $\alpha$.
It suggests that the running time of the minimum weight odd cycle calculation for \cref{thm_simple_odd_beta_separation} is proportional to $|E|^{2.6}$ in practice, which is much smaller than any of the theoretical worst-case guarantees mentioned above.


\begin{table}[htpb]
  \caption{%
    Results for low auto-correlation binary sequence problems (part~1).
  }
  \label{tab_auto_correlation1}
  {\small{%
    \begin{center}
    \setlength{\tabcolsep}{1mm}
    \begin{tabular}{rr|rr|r|rr|rr|rrr|rr}
      $N$ & $R$ & $|V|$ & $|E|$ & $\Psr$ & \multicolumn{2}{c|}{\textbf{$\PfrOne$}} & \multicolumn{2}{c|}{\textbf{$\Pfr$}} & \multicolumn{3}{c|}{\textbf{$\Pcr$}} & \multicolumn{2}{c}{SCIP} \\
      & & & & \textbf{time} & \textbf{gap} & \textbf{time} & \textbf{gap} & \textbf{time} & \textbf{gap} & \textbf{time} & \textbf{cycles} & \textbf{gap} & \textbf{time} \\ \hline
  15 & 2 & 0 & 0 & \SI{0}{\second} & \SI{100}{\percent} & \SI{0}{\second} & \SI{100}{\percent} & \SI{0}{\second} & \SI{100}{\percent} & \SI{0}{\second} & \num{0} & \SI{100}{\percent} & \SI{0}{\second}\\
  15 & 4 & 15 & 88 & \SI{0}{\second} & \SI{67}{\percent} & \SI{0}{\second} & \SI{67}{\percent} & \SI{0}{\second} & \SI{67}{\percent} & \SI{0}{\second} & \num{0} & \SI{60}{\percent} & \SI{3}{\second}\\
  15 & 6 & 15 & 202 & \SI{0}{\second} & \SI{64}{\percent} & \SI{0}{\second} & \SI{64}{\percent} & \SI{0}{\second} & \SI{75}{\percent} & \SI{0}{\second} & \num{420} & \SI{41}{\percent} & \SI{3}{\second}\\
  15 & 8 & 15 & 356 & \SI{0}{\second} & \SI{64}{\percent} & \SI{0}{\second} & \SI{64}{\percent} & \SI{0}{\second} & \SI{75}{\percent} & \SI{2}{\second} & \num{734} & \SI{26}{\percent} & \SI{1}{\second}\\
  15 & 10 & 15 & 513 & \SI{0}{\second} & \SI{64}{\percent} & \SI{0}{\second} & \SI{64}{\percent} & \SI{0}{\second} & \SI{75}{\percent} & \SI{5}{\second} & \num{1087} & \SI{20}{\percent} & \SI{2}{\second}\\
  15 & 12 & 15 & 642 & \SI{0}{\second} & \SI{64}{\percent} & \SI{0}{\second} & \SI{64}{\percent} & \SI{0}{\second} & \SI{74}{\percent} & \SI{18}{\second} & \num{2172} & \SI{20}{\percent} & \SI{1}{\second}\\
  15 & 14 & 15 & 734 & \SI{0}{\second} & \SI{64}{\percent} & \SI{0}{\second} & \SI{64}{\percent} & \SI{1}{\second} & \SI{74}{\percent} & \SI{15}{\second} & \num{1458} & \SI{16}{\percent} & \SI{0}{\second}\\
  15 & 15 & 15 & 753 & \SI{0}{\second} & \SI{64}{\percent} & \SI{0}{\second} & \SI{64}{\percent} & \SI{1}{\second} & \SI{74}{\percent} & \SI{24}{\second} & \num{1683} & \SI{18}{\percent} & \SI{0}{\second}\\
  20 & 3 & 20 & 18 & \SI{0}{\second} & \SI{100}{\percent} & \SI{0}{\second} & \SI{100}{\percent} & \SI{0}{\second} & \SI{100}{\percent} & \SI{0}{\second} & \num{0} & \SI{100}{\percent} & \SI{0}{\second}\\
  20 & 5 & 20 & 187 & \SI{0}{\second} & \SI{65}{\percent} & \SI{0}{\second} & \SI{65}{\percent} & \SI{0}{\second} & \SI{77}{\percent} & \SI{0}{\second} & \num{549} & \SI{50}{\percent} & \SI{1}{\second}\\
  20 & 8 & 20 & 536 & \SI{0}{\second} & \SI{64}{\percent} & \SI{0}{\second} & \SI{64}{\percent} & \SI{0}{\second} & \SI{75}{\percent} & \SI{4}{\second} & \num{1094} & \SI{22}{\percent} & \SI{1}{\second}\\
  20 & 10 & 20 & 813 & \SI{0}{\second} & \SI{64}{\percent} & \SI{0}{\second} & \SI{64}{\percent} & \SI{1}{\second} & \SI{74}{\percent} & \SI{18}{\second} & \num{1520} & \SI{15}{\percent} & \SI{1}{\second}\\
  20 & 13 & 20 & 1227 & \SI{0}{\second} & \SI{64}{\percent} & \SI{1}{\second} & \SI{64}{\percent} & \SI{2}{\second} & \SI{74}{\percent} & \SI{105}{\second} & \num{4544} & \SI{14}{\percent} & \SI{0}{\second}\\
  20 & 15 & 20 & 1474 & \SI{0}{\second} & \SI{63}{\percent} & \SI{1}{\second} & \SI{63}{\percent} & \SI{4}{\second} & \SI{74}{\percent} & \SI{115}{\second} & \num{2954} & \SI{12}{\percent} & \SI{0}{\second}\\
  20 & 18 & 20 & 1757 & \SI{0}{\second} & \SI{63}{\percent} & \SI{1}{\second} & \SI{63}{\percent} & \SI{8}{\second} & \SI{74}{\percent} & \SI{189}{\second} & \num{3577} & \SI{11}{\percent} & \SI{1}{\second}\\
  20 & 20 & 20 & 1839 & \SI{0}{\second} & \SI{63}{\percent} & \SI{2}{\second} & \SI{63}{\percent} & \SI{6}{\second} & \SI{74}{\percent} & \SI{223}{\second} & \num{3696} & \SI{12}{\percent} & \SI{1}{\second}\\
  25 & 4 & 25 & 158 & \SI{0}{\second} & \SI{67}{\percent} & \SI{0}{\second} & \SI{67}{\percent} & \SI{0}{\second} & \SI{67}{\percent} & \SI{0}{\second} & \num{0} & \SI{56}{\percent} & \SI{3}{\second}\\
  25 & 7 & 25 & 536 & \SI{0}{\second} & \SI{64}{\percent} & \SI{0}{\second} & \SI{64}{\percent} & \SI{0}{\second} & \SI{75}{\percent} & \SI{3}{\second} & \num{1121} & \SI{24}{\percent} & \SI{2}{\second}\\
  25 & 10 & 25 & 1113 & \SI{0}{\second} & \SI{64}{\percent} & \SI{1}{\second} & \SI{64}{\percent} & \SI{1}{\second} & \SI{74}{\percent} & \SI{31}{\second} & \num{2387} & \SI{17}{\percent} & \SI{1}{\second}\\
  25 & 13 & 25 & 1757 & \SI{0}{\second} & \SI{64}{\percent} & \SI{2}{\second} & \SI{64}{\percent} & \SI{5}{\second} & \SI{74}{\percent} & \SI{155}{\second} & \num{3490} & \SI{10}{\percent} & \SI{1}{\second}\\
  25 & 16 & 25 & 2429 & \SI{0}{\second} & \SI{63}{\percent} & \SI{3}{\second} & \SI{63}{\percent} & \SI{10}{\second} & \SI{74}{\percent} & \SI{443}{\second} & \num{6609} & \SI{9}{\percent} & \SI{1}{\second}\\
  25 & 19 & 25 & 3015 & \SI{1}{\second} & \SI{63}{\percent} & \SI{5}{\second} & \SI{63}{\percent} & \SI{18}{\second} & \SI{74}{\percent} & \SI{774}{\second} & \num{6064} & \SI{8}{\percent} & \SI{2}{\second}\\
  25 & 22 & 25 & 3455 & \SI{1}{\second} & \SI{63}{\percent} & \SI{6}{\second} & \SI{63}{\percent} & \SI{15}{\second} & \SI{74}{\percent} & \SI{1129}{\second} & \num{7081} & \SI{8}{\percent} & \SI{2}{\second}\\
  25 & 25 & 25 & 3652 & \SI{1}{\second} & \SI{63}{\percent} & \SI{7}{\second} & \SI{63}{\percent} & \SI{28}{\second} & \SI{74}{\percent} & \SI{1369}{\second} & \num{7469} & \SI{9}{\percent} & \SI{2}{\second}
    \end{tabular}
  \end{center}
  }}%
\end{table}


\begin{table}[htpb]
  \caption{%
    Results for low auto-correlation binary sequence problems (part~2).
  }
  \label{tab_auto_correlation2}
  {\small{%
    \begin{center}
    \setlength{\tabcolsep}{1mm}
    \begin{tabular}{rr|rr|r|rr|rr|rrr|rr}
      $N$ & $R$ & $|V|$ & $|E|$ & $\Psr$ & \multicolumn{2}{c|}{\textbf{$\PfrOne$}} & \multicolumn{2}{c|}{\textbf{$\Pfr$}} & \multicolumn{3}{c|}{\textbf{$\Pcr$}} & \multicolumn{2}{c}{SCIP} \\
      & & & & \textbf{time} & \textbf{gap} & \textbf{time} & \textbf{gap} & \textbf{time} & \textbf{gap} & \textbf{time} & \textbf{cycles} & \textbf{gap} & \textbf{time} \\ \hline
  30 & 4 & 30 & 193 & \SI{0}{\second} & \SI{67}{\percent} & \SI{0}{\second} & \SI{67}{\percent} & \SI{0}{\second} & \SI{67}{\percent} & \SI{0}{\second} & \num{0} & \SI{59}{\percent} & \SI{5}{\second}\\
  30 & 8 & 30 & 896 & \SI{0}{\second} & \SI{64}{\percent} & \SI{0}{\second} & \SI{64}{\percent} & \SI{1}{\second} & \SI{75}{\percent} & \SI{11}{\second} & \num{1892} & \SI{23}{\percent} & \SI{1}{\second}\\
  30 & 12 & 30 & 1977 & \SI{0}{\second} & \SI{64}{\percent} & \SI{3}{\second} & \SI{64}{\percent} & \SI{5}{\second} & \SI{74}{\percent} & \SI{208}{\second} & \num{3774} & \SI{10}{\percent} & \SI{1}{\second}\\
  30 & 15 & 30 & 2914 & \SI{1}{\second} & \SI{63}{\percent} & \SI{7}{\second} & \SI{63}{\percent} & \SI{18}{\second} & \SI{74}{\percent} & \SI{738}{\second} & \num{8109} & \SI{8}{\percent} & \SI{2}{\second}\\
  30 & 19 & 30 & 4215 & \SI{1}{\second} & \SI{63}{\percent} & \SI{10}{\second} & \SI{63}{\percent} & \SI{23}{\second} & \SI{74}{\percent} & \SI{1980}{\second} & \num{10486} & \SI{6}{\percent} & \SI{3}{\second}\\
  30 & 23 & 30 & 5346 & \SI{2}{\second} & \SI{63}{\percent} & \SI{16}{\second} & \SI{63}{\percent} & \SI{38}{\second} & \SI{72}{\percent} & \timeout & \num{17361} & \SI{5}{\percent} & \SI{3}{\second}\\
  30 & 27 & 30 & 6133 & \SI{3}{\second} & \SI{63}{\percent} & \SI{29}{\second} & \SI{63}{\percent} & \SI{98}{\second} & \SI{72}{\percent} & \timeout & \num{12346} & \SI{5}{\percent} & \SI{4}{\second}\\
  30 & 30 & 30 & 6382 & \SI{3}{\second} & \SI{63}{\percent} & \SI{27}{\second} & \SI{63}{\percent} & \SI{103}{\second} & \SI{74}{\percent} & \timeout & \num{8173} & \SI{6}{\percent} & \SI{4}{\second}\\
  35 & 5 & 35 & 352 & \SI{0}{\second} & \SI{65}{\percent} & \SI{0}{\second} & \SI{65}{\percent} & \SI{0}{\second} & \SI{77}{\percent} & \SI{1}{\second} & \num{709} & \SI{41}{\percent} & \SI{6}{\second}\\
  35 & 9 & 35 & 1346 & \SI{0}{\second} & \SI{64}{\percent} & \SI{1}{\second} & \SI{64}{\percent} & \SI{1}{\second} & \SI{75}{\percent} & \SI{33}{\second} & \num{2938} & \SI{16}{\percent} & \SI{1}{\second}\\
  35 & 14 & 35 & 3235 & \SI{1}{\second} & \SI{63}{\percent} & \SI{7}{\second} & \SI{63}{\percent} & \SI{16}{\second} & \SI{74}{\percent} & \SI{708}{\second} & \num{9130} & \SI{8}{\percent} & \SI{2}{\second}\\
  35 & 18 & 35 & 4967 & \SI{2}{\second} & \SI{63}{\percent} & \SI{18}{\second} & \SI{63}{\percent} & \SI{47}{\second} & \SI{74}{\percent} & \SI{3031}{\second} & \num{15969} & \SI{4}{\percent} & \SI{3}{\second}\\
  35 & 22 & 35 & 6738 & \SI{3}{\second} & \SI{63}{\percent} & \SI{32}{\second} & \SI{63}{\percent} & \SI{96}{\second} & \SI{72}{\percent} & \timeout & \num{13426} & \SI{4}{\percent} & \SI{5}{\second}\\
  35 & 27 & 35 & 8645 & \SI{4}{\second} & \SI{63}{\percent} & \SI{70}{\second} & \SI{63}{\percent} & \SI{132}{\second} & \SI{69}{\percent} & \timeout & \num{13302} & \SI{4}{\percent} & \SI{7}{\second}\\
  35 & 31 & 35 & 9744 & \SI{5}{\second} & \SI{63}{\percent} & \SI{70}{\second} & \SI{63}{\percent} & \SI{221}{\second} & \SI{63}{\percent} & \timeout & \num{6219} & \SI{3}{\percent} & \SI{8}{\second}\\
  35 & 35 & 35 & 10217 & \SI{6}{\second} & \SI{63}{\percent} & \SI{169}{\second} & \SI{63}{\percent} & \SI{250}{\second} & \SI{63}{\percent} & \timeout & \num{6529} & \SI{5}{\percent} & \SI{9}{\second}\\
  40 & 5 & 40 & 407 & \SI{0}{\second} & \SI{65}{\percent} & \SI{0}{\second} & \SI{65}{\percent} & \SI{0}{\second} & \SI{77}{\percent} & \SI{1}{\second} & \num{700} & \SI{43}{\percent} & \SI{4}{\second}\\
  40 & 10 & 40 & 2013 & \SI{0}{\second} & \SI{64}{\percent} & \SI{2}{\second} & \SI{64}{\percent} & \SI{3}{\second} & \SI{74}{\percent} & \SI{124}{\second} & \num{4331} & \SI{13}{\percent} & \SI{1}{\second}\\
  40 & 15 & 40 & 4354 & \SI{1}{\second} & \SI{63}{\percent} & \SI{13}{\second} & \SI{63}{\percent} & \SI{30}{\second} & \SI{74}{\percent} & \SI{1118}{\second} & \num{8503} & \SI{8}{\percent} & \SI{3}{\second}\\
  40 & 20 & 40 & 7203 & \SI{3}{\second} & \SI{63}{\percent} & \SI{42}{\second} & \SI{63}{\percent} & \SI{105}{\second} & \SI{72}{\percent} & \timeout & \num{14235} & \SI{3}{\percent} & \SI{5}{\second}\\
  40 & 25 & 40 & 10082 & \SI{6}{\second} & \SI{63}{\percent} & \SI{108}{\second} & \SI{63}{\percent} & \SI{179}{\second} & \SI{63}{\percent} & \SI{3850}{\second} & \num{6404} & \SI{3}{\percent} & \SI{8}{\second}\\
  40 & 30 & 40 & 12650 & \SI{9}{\second} & \SI{62}{\percent} & \SI{208}{\second} & \SI{62}{\percent} & \SI{455}{\second} & \SI{62}{\percent} & \timeout & \num{0} & \SI{2}{\percent} & \SI{13}{\second}\\
  40 & 35 & 40 & 14533 & \SI{16}{\second} & \SI{62}{\percent} & \SI{366}{\second} & \SI{62}{\percent} & \SI{744}{\second} & \SI{62}{\percent} & \timeout & \num{0} & \SI{2}{\percent} & \SI{18}{\second}\\
  40 & 40 & 40 & 15344 & \SI{16}{\second} & \SI{62}{\percent} & \SI{276}{\second} & \SI{62}{\percent} & \SI{659}{\second} & \SI{62}{\percent} & \timeout & \num{0} & \SI{4}{\percent} & \SI{22}{\second}
    \end{tabular}
  \end{center}
  }}%
\end{table}


\begin{table}[htpb]
  \caption{%
    Results for low auto-correlation binary sequence problems (part~3).
  }
  \label{tab_auto_correlation3}
  {\small{%
    \begin{center}
    \setlength{\tabcolsep}{1mm}
    \begin{tabular}{rr|rr|r|rr|rr|rrr|rr}
      $N$ & $R$ & $|V|$ & $|E|$ & $\Psr$ & \multicolumn{2}{c|}{\textbf{$\PfrOne$}} & \multicolumn{2}{c|}{\textbf{$\Pfr$}} & \multicolumn{3}{c|}{\textbf{$\Pcr$}} & \multicolumn{2}{c}{SCIP} \\
      & & & & \textbf{time} & \textbf{gap} & \textbf{time} & \textbf{gap} & \textbf{time} & \textbf{gap} & \textbf{time} & \textbf{cycles} & \textbf{gap} & \textbf{time} \\ \hline
  45 & 6 & 45 & 742 & \SI{0}{\second} & \SI{64}{\percent} & \SI{0}{\second} & \SI{64}{\percent} & \SI{0}{\second} & \SI{75}{\percent} & \SI{3}{\second} & \num{1202} & \SI{29}{\percent} & \SI{1}{\second}\\
  45 & 12 & 45 & 3312 & \SI{1}{\second} & \SI{63}{\percent} & \SI{7}{\second} & \SI{63}{\percent} & \SI{10}{\second} & \SI{74}{\percent} & \SI{488}{\second} & \num{6635} & \SI{6}{\percent} & \SI{2}{\second}\\
  45 & 17 & 45 & 6407 & \SI{3}{\second} & \SI{63}{\percent} & \SI{33}{\second} & \SI{63}{\percent} & \SI{70}{\second} & \SI{72}{\percent} & \timeout & \num{18016} & \SI{5}{\percent} & \SI{5}{\second}\\
  45 & 23 & 45 & 10731 & \SI{7}{\second} & \SI{63}{\percent} & \SI{141}{\second} & \SI{63}{\percent} & \SI{284}{\second} & \SI{63}{\percent} & \timeout & \num{6794} & \SI{2}{\percent} & \SI{11}{\second}\\
  45 & 29 & 45 & 15092 & \SI{17}{\second} & \SI{62}{\percent} & \SI{382}{\second} & \SI{62}{\percent} & \SI{742}{\second} & \SI{62}{\percent} & \timeout & \num{0} & \SI{2}{\percent} & \SI{20}{\second}\\
  45 & 34 & 45 & 18303 & \SI{21}{\second} & \SI{63}{\percent} & \SI{559}{\second} & \SI{63}{\percent} & \SI{1088}{\second} & \SI{63}{\percent} & \timeout & \num{0} & \SI{2}{\percent} & \SI{26}{\second}\\
  45 & 40 & 45 & 21033 & \SI{32}{\second} & \SI{61}{\percent} & \SI{617}{\second} & \SI{61}{\percent} & \SI{1335}{\second} & \SI{61}{\percent} & \timeout & \num{0} & \SI{1}{\percent} & \SI{40}{\second}\\
  45 & 45 & 45 & 21948 & \SI{36}{\second} & \SI{63}{\percent} & \SI{675}{\second} & \SI{63}{\percent} & \SI{1066}{\second} & \SI{63}{\percent} & \timeout & \num{0} & \SI{2}{\percent} & \SI{44}{\second}\\
  50 & 7 & 50 & 1186 & \SI{0}{\second} & \SI{64}{\percent} & \SI{1}{\second} & \SI{64}{\percent} & \SI{1}{\second} & \SI{75}{\percent} & \SI{8}{\second} & \num{1524} & \SI{22}{\percent} & \SI{2}{\second}\\
  50 & 13 & 50 & 4407 & \SI{1}{\second} & \SI{63}{\percent} & \SI{12}{\second} & \SI{63}{\percent} & \SI{17}{\second} & \SI{74}{\percent} & \SI{1047}{\second} & \num{8590} & \SI{8}{\percent} & \SI{3}{\second}\\
  50 & 19 & 50 & 9015 & \SI{5}{\second} & \SI{63}{\percent} & \SI{68}{\second} & \SI{63}{\percent} & \SI{145}{\second} & \SI{70}{\percent} & \timeout & \num{14068} & \SI{4}{\percent} & \SI{8}{\second}\\
  50 & 25 & 50 & 14362 & \SI{17}{\second} & \SI{61}{\percent} & \SI{274}{\second} & \SI{61}{\percent} & \SI{542}{\second} & \SI{61}{\percent} & \timeout & \num{0} & \SI{2}{\percent} & \SI{19}{\second}\\
  50 & 32 & 50 & 20674 & \SI{30}{\second} & \SI{61}{\percent} & \SI{806}{\second} & \SI{61}{\percent} & \SI{1495}{\second} & \SI{61}{\percent} & \timeout & \num{0} & \SI{1}{\percent} & \SI{38}{\second}\\
  50 & 38 & 50 & 25396 & \SI{43}{\second} & \SI{61}{\percent} & \SI{1205}{\second} & \SI{61}{\percent} & \SI{2222}{\second} & \SI{61}{\percent} & \timeout & \num{0} & \SI{2}{\percent} & \SI{57}{\second}\\
  50 & 44 & 50 & 28797 & \SI{63}{\second} & \SI{61}{\percent} & \SI{1762}{\second} & \SI{61}{\percent} & \SI{3091}{\second} & \SI{61}{\percent} & \timeout & \num{0} & \SI{1}{\percent} & \SI{77}{\second}\\
  50 & 50 & 50 & 30221 & \SI{83}{\second} & \SI{60}{\percent} & \SI{1828}{\second} & \SI{60}{\percent} & \SI{2584}{\second} & \SI{60}{\percent} & \timeout & \num{0} & \SI{2}{\percent} & \SI{93}{\second}\\
  55 & 7 & 55 & 1316 & \SI{0}{\second} & \SI{64}{\percent} & \SI{1}{\second} & \SI{64}{\percent} & \SI{1}{\second} & \SI{75}{\percent} & \SI{9}{\second} & \num{1646} & \SI{19}{\percent} & \SI{1}{\second}\\
  55 & 14 & 55 & 5735 & \SI{2}{\second} & \SI{63}{\percent} & \SI{20}{\second} & \SI{63}{\percent} & \SI{38}{\second} & \SI{74}{\percent} & \SI{2688}{\second} & \num{16754} & \SI{4}{\percent} & \SI{3}{\second}\\
  55 & 21 & 55 & 12203 & \SI{10}{\second} & \SI{63}{\percent} & \SI{193}{\second} & \SI{63}{\percent} & \SI{344}{\second} & \SI{63}{\percent} & \timeout & \num{7693} & \SI{2}{\percent} & \SI{13}{\second}\\
  55 & 28 & 55 & 19592 & \SI{32}{\second} & \SI{43}{\percent} & \SI{331}{\second} & \SI{43}{\percent} & \SI{332}{\second} & \SI{43}{\percent} & \SI{388}{\second} & \num{0} & \SI{2}{\percent} & \SI{36}{\second}\\
  55 & 35 & 55 & 27202 & \SI{60}{\second} & \SI{46}{\percent} & \SI{876}{\second} & \SI{46}{\percent} & \SI{878}{\second} & \SI{46}{\percent} & \SI{1029}{\second} & \num{0} & \SI{1}{\percent} & \SI{74}{\second}\\
  55 & 42 & 55 & 33867 & \SI{106}{\second} & \SI{38}{\percent} & \timeout & \SI{38}{\percent} & \timeout & \SI{38}{\percent} & \timeout & \num{0} & \SI{1}{\percent} & \SI{119}{\second}\\
  55 & 49 & 55 & 38515 & \SI{183}{\second} & \SI{54}{\percent} & \SI{3390}{\second} & \SI{54}{\percent} & \SI{3347}{\second} & \SI{54}{\percent} & \SI{3544}{\second} & \num{0} & \SI{1}{\percent} & \SI{154}{\second}\\
  55 & 55 & 55 & 40087 & \SI{172}{\second} & \SI{52}{\percent} & \SI{3118}{\second} & \SI{52}{\percent} & \SI{3063}{\second} & \SI{52}{\percent} & \SI{3581}{\second} & \num{0} & \SI{2}{\percent} & \SI{195}{\second}
    \end{tabular}
  \end{center}
  }}%
\end{table}


\begin{landscape}

\begin{table}[htpb]
  \caption{%
    Results for random high-degree instances.
    Integrality gaps and computation times are averaged over 5 random instances.
  }
  \label{tab_high_degree}
  {\footnotesize{%
  \begin{center}
    \setlength{\tabcolsep}{2mm}
    \begin{tabular}{rrr|rr|r|rr|rr|rrr|rr}
      $n$ & $m$ & $d$ & $|V|$ & $|E|$ & $\Psr$ & \multicolumn{2}{c|}{\textbf{$\PfrOne$}} & \multicolumn{2}{c|}{\textbf{$\Pfr$}} & \multicolumn{3}{c|}{\textbf{$\Pcr$}} & \multicolumn{2}{c}{SCIP} \\
      & & & & & \textbf{time} & \textbf{gap} & \textbf{time} & \textbf{gap} & \textbf{time} & \textbf{gap} & \textbf{time} & \textbf{cycles} & \textbf{gap} & \textbf{time} \\ \hline
  50 & 250 & 2 & 50 & 237.2 & \SI{0}{\second} & \SI{44}{\percent} & \SI{0}{\second} & \SI{45}{\percent} & \SI{0}{\second} & \SI{98}{\percent} & \SI{1}{\second} & \num{2748} & \SI{55}{\percent} & \SI{2}{\second}\\
  50 & 250 & 3 & 50 & 240.0 & \SI{0}{\second} & \SI{45}{\percent} & \SI{0}{\second} & \SI{54}{\percent} & \SI{0}{\second} & \SI{85}{\percent} & \SI{3}{\second} & \num{3291} & \SI{20}{\percent} & \SI{3}{\second}\\
  50 & 250 & 4 & 50 & 238.4 & \SI{0}{\second} & \SI{32}{\percent} & \SI{0}{\second} & \SI{57}{\percent} & \SI{0}{\second} & \SI{73}{\percent} & \SI{4}{\second} & \num{2893} & \SI{9}{\percent} & \SI{3}{\second}\\
  50 & 500 & 2 & 50 & 475.2 & \SI{0}{\second} & \SI{41}{\percent} & \SI{0}{\second} & \SI{41}{\percent} & \SI{0}{\second} & \SI{88}{\percent} & \SI{13}{\second} & \num{7287} & \SI{26}{\percent} & \SI{3}{\second}\\
  50 & 500 & 3 & 50 & 474.4 & \SI{0}{\second} & \SI{43}{\percent} & \SI{0}{\second} & \SI{47}{\percent} & \SI{0}{\second} & \SI{81}{\percent} & \SI{23}{\second} & \num{8610} & \SI{9}{\percent} & \SI{3}{\second}\\
  50 & 500 & 4 & 50 & 475.2 & \SI{0}{\second} & \SI{30}{\percent} & \SI{0}{\second} & \SI{62}{\percent} & \SI{1}{\second} & \SI{76}{\percent} & \SI{24}{\second} & \num{6309} & \SI{4}{\percent} & \SI{3}{\second}\\
  50 & 1000 & 2 & 50 & 955.6 & \SI{0}{\second} & \SI{44}{\percent} & \SI{1}{\second} & \SI{44}{\percent} & \SI{1}{\second} & \SI{84}{\percent} & \SI{127}{\second} & \num{17028} & \SI{13}{\percent} & \SI{2}{\second}\\
  50 & 1000 & 3 & 50 & 950.0 & \SI{0}{\second} & \SI{41}{\percent} & \SI{1}{\second} & \SI{43}{\percent} & \SI{1}{\second} & \SI{79}{\percent} & \SI{187}{\second} & \num{16431} & \SI{5}{\percent} & \SI{2}{\second}\\
  50 & 1000 & 4 & 50 & 951.0 & \SI{0}{\second} & \SI{29}{\percent} & \SI{1}{\second} & \SI{55}{\percent} & \SI{2}{\second} & \SI{55}{\percent} & \SI{2}{\second} & \num{0} & \SI{2}{\percent} & \SI{2}{\second}\\
  100 & 500 & 2 & 100 & 476.2 & \SI{0}{\second} & \SI{30}{\percent} & \SI{0}{\second} & \SI{31}{\percent} & \SI{0}{\second} & \SI{86}{\percent} & \SI{6}{\second} & \num{5491} & \SI{43}{\percent} & \SI{2}{\second}\\
  100 & 500 & 3 & 100 & 478.4 & \SI{0}{\second} & \SI{29}{\percent} & \SI{0}{\second} & \SI{31}{\percent} & \SI{0}{\second} & \SI{67}{\percent} & \SI{13}{\second} & \num{6758} & \SI{9}{\percent} & \SI{2}{\second}\\
  100 & 500 & 4 & 100 & 473.2 & \SI{0}{\second} & \SI{25}{\percent} & \SI{0}{\second} & \SI{39}{\percent} & \SI{0}{\second} & \SI{44}{\percent} & \SI{6}{\second} & \num{2008} & \SI{4}{\percent} & \SI{2}{\second}\\
  100 & 1000 & 2 & 100 & 952.0 & \SI{0}{\second} & \SI{25}{\percent} & \SI{1}{\second} & \SI{26}{\percent} & \SI{1}{\second} & \SI{75}{\percent} & \SI{68}{\second} & \num{16234} & \SI{16}{\percent} & \SI{6}{\second}\\
  100 & 1000 & 3 & 100 & 941.8 & \SI{0}{\second} & \SI{28}{\percent} & \SI{1}{\second} & \SI{31}{\percent} & \SI{1}{\second} & \SI{58}{\percent} & \SI{103}{\second} & \num{16259} & \SI{3}{\percent} & \SI{4}{\second}\\
  100 & 1000 & 4 & 100 & 954.0 & \SI{0}{\second} & \SI{19}{\percent} & \SI{1}{\second} & \SI{35}{\percent} & \SI{2}{\second} & \SI{50}{\percent} & \SI{138}{\second} & \num{16464} & \SI{1}{\percent} & \SI{4}{\second}\\
  100 & 2000 & 2 & 100 & 1903.6 & \SI{1}{\second} & \SI{28}{\percent} & \SI{2}{\second} & \SI{28}{\percent} & \SI{2}{\second} & \SI{72}{\percent} & \SI{796}{\second} & \num{41560} & \SI{9}{\percent} & \SI{6}{\second}\\
  100 & 2000 & 3 & 100 & 1905.0 & \SI{1}{\second} & \SI{32}{\percent} & \SI{4}{\second} & \SI{35}{\percent} & \SI{4}{\second} & \SI{60}{\percent} & \SI{1341}{\second} & \num{52987} & \SI{2}{\percent} & \SI{6}{\second}\\
  100 & 2000 & 4 & 100 & 1902.8 & \SI{1}{\second} & \SI{19}{\percent} & \SI{3}{\second} & \SI{45}{\percent} & \SI{19}{\second} & \SI{57}{\percent} & \SI{1275}{\second} & \num{39995} & \SI{1}{\percent} & \SI{6}{\second}\\
  200 & 1000 & 2 & 200 & 949.0 & \SI{1}{\second} & \SI{11}{\percent} & \SI{1}{\second} & \SI{11}{\percent} & \SI{1}{\second} & \SI{73}{\percent} & \SI{35}{\second} & \num{14059} & \SI{27}{\percent} & \SI{2}{\second}\\
  200 & 1000 & 3 & 200 & 954.4 & \SI{1}{\second} & \SI{15}{\percent} & \SI{1}{\second} & \SI{16}{\percent} & \SI{1}{\second} & \SI{46}{\percent} & \SI{55}{\second} & \num{14004} & \SI{4}{\percent} & \SI{4}{\second}\\
  200 & 1000 & 4 & 200 & 950.2 & \SI{1}{\second} & \SI{18}{\percent} & \SI{2}{\second} & \SI{20}{\percent} & \SI{2}{\second} & \SI{37}{\percent} & \SI{86}{\second} & \num{12595} & \SI{1}{\percent} & \SI{3}{\second}\\
  200 & 2000 & 2 & 200 & 1903.8 & \SI{3}{\second} & \SI{12}{\percent} & \SI{3}{\second} & \SI{12}{\percent} & \SI{3}{\second} & \SI{58}{\percent} & \SI{400}{\second} & \num{36450} & \SI{9}{\percent} & \SI{6}{\second}\\
  200 & 2000 & 3 & 200 & 1903.8 & \SI{3}{\second} & \SI{17}{\percent} & \SI{4}{\second} & \SI{18}{\percent} & \SI{4}{\second} & \SI{44}{\percent} & \SI{719}{\second} & \num{39039} & \SI{1}{\percent} & \SI{8}{\second}\\
  200 & 2000 & 4 & 200 & 1911.0 & \SI{2}{\second} & \SI{15}{\percent} & \SI{5}{\second} & \SI{22}{\percent} & \SI{8}{\second} & \SI{23}{\percent} & \SI{48}{\second} & \num{2614} & \SI{0}{\percent} & \SI{8}{\second}\\
  200 & 4000 & 2 & 200 & 3806.0 & \SI{9}{\second} & \SI{15}{\percent} & \SI{12}{\second} & \SI{15}{\percent} & \SI{12}{\second} & \SI{48}{\percent} & \timeout & \num{90481} & \SI{5}{\percent} & \SI{22}{\second}\\
  200 & 4000 & 3 & 200 & 3809.4 & \SI{9}{\second} & \SI{21}{\percent} & \SI{19}{\second} & \SI{23}{\percent} & \SI{21}{\second} & \SI{37}{\percent} & \timeout & \num{64306} & \SI{1}{\percent} & \SI{19}{\second}\\
  200 & 4000 & 4 & 200 & 3812.0 & \SI{8}{\second} & \SI{14}{\percent} & \SI{18}{\second} & \SI{27}{\percent} & \SI{51}{\second} & \SI{31}{\percent} & \SI{2714}{\second} & \num{46141} & \SI{0}{\percent} & \SI{13}{\second}
    \end{tabular}
  \end{center}
  }}%
\end{table}

\end{landscape}

\begin{figure}
  \begin{center}
    \begin{tikzpicture}
      \begin{semilogyaxis}[
          grid,
          xlabel={Exponent $\alpha$ of estimation $t_i \sim |E_i|^\alpha$},
          ylabel={Ratio $\max \left( t_i \cdot |E_i|^{-\alpha} \right) / \min \left( t_i \cdot |E_i|^{-\alpha} \right)$},
          ymin=9,
          width=0.9\textwidth,
          height=0.7\textwidth,
        ]
        
        \addplot[thin] coordinates {
          (1.0,246.115) (1.01,241.000) (1.02,235.992) (1.03,231.088) (1.04,226.286) (1.05,221.584) (1.06,216.979) (1.07,212.470) (1.08,208.055) (1.09,203.732)
          (1.1,199.498) (1.11,195.353) (1.12,191.293) (1.13,187.318) (1.14,183.425) (1.15,179.614) (1.16,175.881) (1.17,172.226) (1.18,168.647) (1.19,165.143)
          (1.2,161.711) (1.21,158.351) (1.22,155.060) (1.23,151.838) (1.24,148.683) (1.25,145.593) (1.26,142.568) (1.27,139.605) (1.28,136.704) (1.29,133.863)
          (1.3,131.081) (1.31,128.358) (1.32,125.690) (1.33,123.078) (1.34,120.521) (1.35,118.016) (1.36,115.564) (1.37,113.162) (1.38,110.811) (1.39,108.508)
          (1.4,106.253) (1.41,104.045) (1.42,101.883) (1.43,99.766) (1.44,97.693) (1.45,95.663) (1.46,93.675) (1.47,91.728) (1.48,89.822) (1.49,87.956)
          (1.5,86.128) (1.51,84.338) (1.52,82.585) (1.53,80.869) (1.54,79.189) (1.55,77.543) (1.56,75.932) (1.57,74.354) (1.58,72.809) (1.59,71.296)
          (1.6,69.814) (1.61,68.364) (1.62,66.943) (1.63,65.552) (1.64,64.239) (1.65,63.341) (1.66,62.456) (1.67,61.584) (1.68,60.724) (1.69,59.875)
          (1.7,59.039) (1.71,58.214) (1.72,57.401) (1.73,56.599) (1.74,55.808) (1.75,55.029) (1.76,54.260) (1.77,53.502) (1.78,52.755) (1.79,52.018)
          (1.8,51.291) (1.81,50.575) (1.82,49.868) (1.83,49.172) (1.84,48.485) (1.85,47.807) (1.86,47.140) (1.87,46.481) (1.88,45.832) (1.89,45.191)
          (1.9,44.560) (1.91,43.938) (1.92,43.324) (1.93,42.719) (1.94,42.122) (1.95,41.534) (1.96,40.953) (1.97,40.381) (1.98,39.817) (1.99,39.261)
          (2.0,38.712) (2.01,38.172) (2.02,37.638) (2.03,37.113) (2.04,36.594) (2.05,36.083) (2.06,35.579) (2.07,35.082) (2.08,34.592) (2.09,34.109)
          (2.1,33.632) (2.11,33.162) (2.12,32.699) (2.13,32.260) (2.14,32.034) (2.15,31.810) (2.16,31.588) (2.17,31.368) (2.18,31.148) (2.19,30.931)
          (2.2,30.715) (2.21,30.500) (2.22,30.287) (2.23,30.075) (2.24,29.865) (2.25,29.657) (2.26,29.450) (2.27,29.244) (2.28,29.039) (2.29,28.837)
          (2.3,28.635) (2.31,28.435) (2.32,28.236) (2.33,28.039) (2.34,27.843) (2.35,27.649) (2.36,27.456) (2.37,27.264) (2.38,27.073) (2.39,26.884)
          (2.4,26.696) (2.41,26.510) (2.42,26.325) (2.43,26.141) (2.44,25.958) (2.45,25.777) (2.46,25.597) (2.47,25.418) (2.48,25.240) (2.49,25.064)
          (2.5,24.889) (2.51,24.715) (2.52,24.554) (2.53,24.552) (2.54,24.550) (2.55,24.547) (2.56,24.545) (2.57,24.543) (2.58,24.541) (2.59,24.539)
          (2.6,24.537) (2.61,24.535) (2.62,24.533) (2.63,24.531) (2.64,24.529) (2.65,24.527) (2.66,24.525) (2.67,24.635) (2.68,24.803) (2.69,25.053)
          (2.7,25.397) (2.71,25.747) (2.72,26.102) (2.73,26.461) (2.74,26.825) (2.75,27.195) (2.76,27.569) (2.77,27.948) (2.78,28.333) (2.79,28.723)
          (2.8,29.119) (2.81,29.520) (2.82,29.926) (2.83,30.338) (2.84,30.756) (2.85,31.179) (2.86,31.609) (2.87,32.044) (2.88,32.485) (2.89,32.932)
          (2.9,33.385) (2.91,33.845) (2.92,34.311) (2.93,34.783) (2.94,35.262) (2.95,35.748) (2.96,36.240) (2.97,36.739) (2.98,37.245) (2.99,37.757)
          (3.0,38.277) (3.01,38.804) (3.02,39.339) (3.03,39.880) (3.04,40.429) (3.05,40.986) (3.06,41.550) (3.07,42.122) (3.08,42.702) (3.09,43.290)
          (3.1,43.886) (3.11,44.490) (3.12,45.103) (3.13,45.724) (3.14,46.353) (3.15,46.991) (3.16,47.638) (3.17,48.294) (3.18,49.122) (3.19,50.147)
          (3.2,51.194) (3.21,52.263) (3.22,53.354) (3.23,54.468) (3.24,55.605) (3.25,56.766) (3.26,57.952) (3.27,59.162) (3.28,60.397) (3.29,61.658)
          (3.3,62.945) (3.31,64.259) (3.32,65.601) (3.33,66.970) (3.34,68.369) (3.35,69.796) (3.36,71.253) (3.37,72.741) (3.38,74.260) (3.39,75.810)
          (3.4,77.393) (3.41,79.009) (3.42,80.658) (3.43,82.342) (3.44,84.062) (3.45,85.817) (3.46,87.608) (3.47,89.437) (3.48,91.305) (3.49,93.211)
          (3.5,95.157) (3.51,97.144) (3.52,99.172) (3.53,101.243) (3.54,103.356) (3.55,105.514) (3.56,107.717) (3.57,109.966) (3.58,112.262) (3.59,114.606)
          (3.6,117.487) (3.61,120.784) (3.62,124.174) (3.63,127.658) (3.64,131.241) (3.65,134.924) (3.66,138.710) (3.67,142.603) (3.68,146.604) (3.69,150.719)
          (3.7,154.948) (3.71,159.296) (3.72,163.767) (3.73,168.363) (3.74,173.087) (3.75,177.945) (3.76,182.938) (3.77,188.072) (3.78,193.350) (3.79,198.776)
          (3.8,204.354) (3.81,210.089) (3.82,215.984) (3.83,222.046) (3.84,228.277) (3.85,234.683) (3.86,241.272) (3.87,248.085) (3.88,255.090) (3.89,262.292)
          (3.9,269.698) (3.91,277.313) (3.92,285.143) (3.93,293.195) (3.94,301.473) (3.95,309.985) (3.96,318.738) (3.97,327.738) (3.98,336.992) (3.99,346.507)
          (4.0,356.291) (4.01,366.351) (4.02,376.695) (4.03,387.331) (4.04,398.268) (4.05,409.513) (4.06,421.076) (4.07,432.965) (4.08,445.190) (4.09,457.761)
          (4.1,470.686) (4.11,483.976) (4.12,497.641) (4.13,511.692) (4.14,526.140) (4.15,540.996) (4.16,556.272) (4.17,571.978) (4.18,588.128) (4.19,604.735)
          (4.2,621.810) (4.21,639.367) (4.22,657.420) (4.23,675.983) (4.24,695.069) (4.25,714.695) (4.26,734.875) (4.27,755.625) (4.28,776.960) (4.29,798.898)
          (4.3,821.455) (4.31,844.650) (4.32,868.499) (4.33,893.022) (4.34,918.237) (4.35,944.164) (4.36,970.823) (4.37,998.234) (4.38,1026.420) (4.39,1055.402)
          (4.4,1085.202) (4.41,1115.843) (4.42,1147.350) (4.43,1179.746) (4.44,1213.057) (4.45,1247.308) (4.46,1282.527) (4.47,1318.740) (4.48,1355.975) (4.49,1394.262)
          (4.5,1433.630) (4.51,1474.109) (4.52,1515.732) (4.53,1558.529) (4.54,1602.535) (4.55,1647.784) (4.56,1694.310) (4.57,1742.150) (4.58,1791.341) (4.59,1841.921)
          (4.6,1893.928) (4.61,1947.405) (4.62,2002.391) (4.63,2058.930) (4.64,2117.065) (4.65,2176.841) (4.66,2238.306) (4.67,2301.506) (4.68,2366.490) (4.69,2433.310)
          (4.7,2502.016) (4.71,2572.662) (4.72,2645.302) (4.73,2719.994) (4.74,2796.795) (4.75,2875.764) (4.76,2956.963) (4.77,3040.455) (4.78,3126.304) (4.79,3214.577)
          (4.8,3305.343) (4.81,3398.672) (4.82,3494.635) (4.83,3593.308) (4.84,3694.768) (4.85,3799.092) (4.86,3906.362) (4.87,4016.660) (4.88,4130.073) (4.89,4246.688)
          (4.9,4366.596) (4.91,4489.890) (4.92,4616.665) (4.93,4747.019) (4.94,4881.054) (4.95,5018.874) (4.96,5160.585) (4.97,5306.298) (4.98,5456.124) (4.99,5610.181)
          (5.0,5768.588)
          };

      \end{semilogyaxis}
    \end{tikzpicture}
  \end{center}
  \caption{%
    Asymptotic behavior of running times of minimum weight odd cycle computations for all random high degree instances.
    Depicted are the relative errors in an estimation of running times $t_i$ as $C \cdot |E_i|^\alpha$ for a suitable constant $C \in \R$ for various exponents $\alpha \in [1,5]$, where $t_i$ is the measured running time in computation $i$, and $|E_i|$ is the corresponding number of hyperedges.
  }
  \label{fig_moc}
\end{figure}
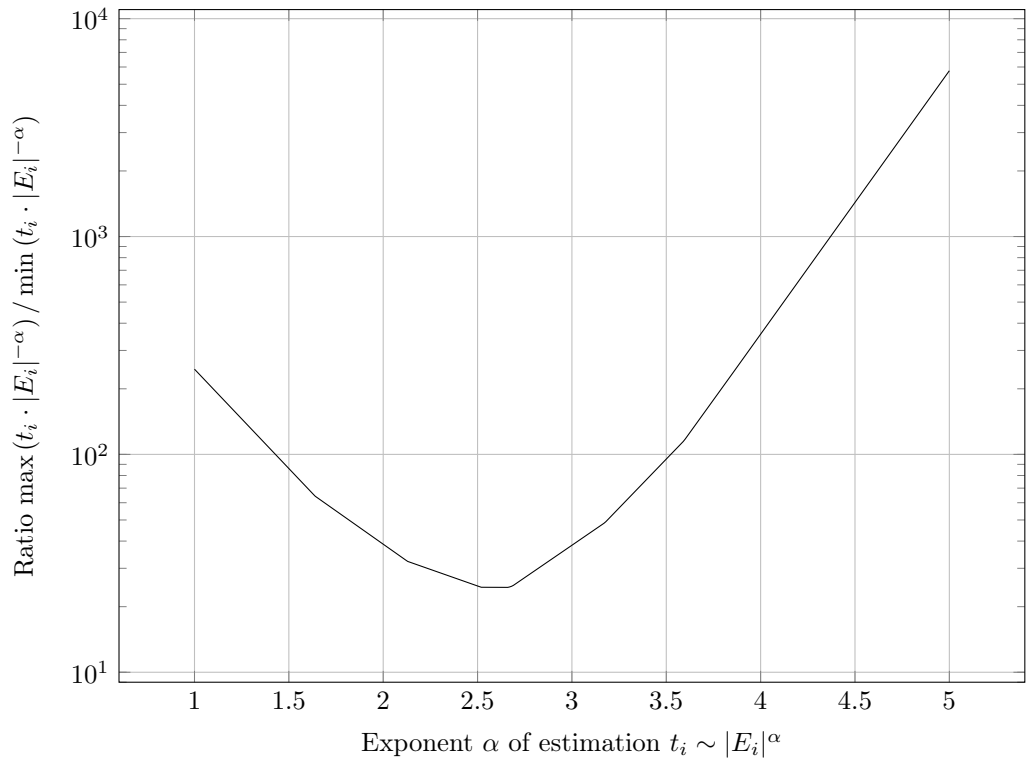

\section{Discussion and conclusions}
\label{sec_conclusions}

We would like to conclude our paper with a reflection of our computational results as well as open questions that could be investigated.

Our computations show that simple odd $\beta$-cycles sometimes close a significant amount of integrality gap, but also sometimes do not strengthen the flower relaxation at all.
In any case, the separation cost is currently prohibitively high.
We believe that this cost can be reduced significantly by pursuing further research.
First, several tricks to speed up the generation of odd cycle inequalities (for maximum cut and stable set) are known.
Examples are the computation of a minimum spanning tree in the auxiliary graph and considering only the fundamental cycles with respect to this tree~\cite{BarahonaGJR88}, the consideration of only a subset of the nodes, or the restriction to primal separation~\cite{LetchfordL03}.
It is worth noting that in the quadratic case, i.e., $|e| = 2$ holds for all $e \in E$, our auxiliary graph has $\orderTheta(|V|+|E|)$ nodes while the auxiliary graph used for maximum cut requires only $\orderTheta(|V|)$ nodes~\cite{BarahonaM86}.
Hence, we believe that the overall algorithm can be improved significantly.

Besides speeding up the computational cost for generating simple odd $\beta$-cycle inequalities, one can also try to strengthen them.
For instance, many of our building block inequalities are the sum of other valid inequalities, e.g., we had to add $x_v \geq 0$ for some nodes $v$ in order to make the coefficient even.
If a node appears in more than one building block (due to overlapping edges), then one can subtract $2x_v \geq 0$ and still obtain a valid (but stronger) cut.
Moreover, one can try to generate strongest-possible simple odd $\beta$-cycle inequalities by replacing those that are already redundant for the multilinear polytope for the support hypergraph.
This was done for the maximum cut problem in~\cite{JuengerM19} by checking cycles for being chordless.
Generally, cycles of short combinatorial length are advantageous for this sake, but also because the resulting inequalities tend to be sparser.
Finally, there is a difference between simple and (non-simple) odd $\beta$-cycle inequalities, and we believe that this might lead to stronger cuts with only little computational cost on top of the separation.

The next research direction has a more theoretical flavor.
The LP relaxations defined by odd-cycle inequalities~\cite{BarahonaM86} for the cut polytope and the affinely isomorphic correlation polytope (see~\cite{DezaL09}) have the following property:
when maximizing a specific objective vector, then one can remove a subset of the odd-cycle inequalities upfront without changing the optimum.
More precisely, the removal is based only on the sign pattern of the objective vector (see Theorem~2 in~\cite{Michini18}).
Since the simple odd $\beta$-cycle inequalities can be seen as an extension of the odd cycle inequalities for the cut polytope, the research question is whether a similar property can be proven for simple odd $\beta$-cycle inequalities.

A final research problem is that of finding alternative linearizations.
Given our insight about the impact of a large number of edges, we propose to investigate means of linearizing a given polynomial with fewer auxiliary variables.
A number of natural questions arise, e.g., what is the smallest number of additional binary variables that one needs to add such that a given polynomial is linear in the original and the auxiliary variables.

\paragraph{Acknowledgements.}
  The authors are grateful to the anonymous reviewers whose comments led to improvements of this manuscript, in particular for the suggestion of analyzing the running times of odd cycle search.
  A.~Del Pia is partially funded by AFOSR grant FA9550-23-1-0433. Any opinions, findings, and conclusions or recommendations expressed in this material are those of the authors and do not necessarily reflect the views of the Air Force Office of Scientific Research.
  M.~Walter acknowledges funding support from the Dutch Research Council (NWO) on grant number OCENW.M20.151.

\bibliographystyle{plainurl}
\bibliography{odd-beta-cycle-separation}

\end{document}